\documentclass[sigconf, nonacm]{acmart}
\usepackage{booktabs}
\usepackage{subcaption}
\usepackage{tabularx} 
\usepackage{siunitx}
\usepackage{algorithm}
\usepackage{algpseudocode}
\usepackage{siunitx}
\newtheorem*{theorem 2.1}{Theorem 2.1}
\newtheorem*{theorem 2.2}{Theorem 2.2}
\newtheorem*{theorem 3.1}{Theorem 3.1}
\newtheorem*{theorem 4.1}{Theorem 4.1}
\newtheorem*{theorem 4.2}{Theorem 4.2}

\begin{document}

\title{Downsizing Diffusion Models for Cardinality Estimation}

\author{Xinhe Mu}
\email{muxinhe22@mails.ucas.ac.cn}
\orcid{0009-0002-3952-5966}
\affiliation{%
  \department{Academy of Mathematics and Systems Sciences}
  \institution{Chinese Academy of Sciences}
  \city{Beijing}
  \country{China}
}
\author{Zhaoqi Zhou}\authornote{Corresponding author.}
\email{zhouzhaoqi1@huawei.com}
\affiliation{%
  \institution{Huawei Technologies Co., Ltd.}
  \city{Beijing}
  \country{China}
}
\author{Zaijiu Shang}
\email{zaijiu@simis.cn; zaijiu@amss.ac.cn
}
\affiliation{%
  \institution{Center for Mathematics
and Interdisciplinary Sciences at Fudan University}
  \institution{Shanghai Institute for Mathematics
and Interdisciplinary Sciences}
  \city{Shanghai}
  \country{China}
}
\author{Chuan Zhou}
\email{zhouchuan@amss.ac.cn
}
\affiliation{%
  \department{Academy of Mathematics and Systems Sciences}
  \institution{Chinese Academy of Sciences}
  \city{Beijing}
  \country{China}
}
\author{Gang Fu}
\email{fugang12@huawei.com}
\affiliation{%
  \institution{Huawei Technologies Co., Ltd.}
  \city{Beijing}
  \country{China}
}
\author{Guiying Yan}
\email{yangy@amt.ac.cn}
\affiliation{%
  \department{Academy of Mathematics and Systems Sciences}
  \institution{Chinese Academy of Sciences}
  \city{Beijing}
  \country{China}
}
\author{Guoliang Li}
\email{liguoliang@tsinghua.edu.cn}
\affiliation{%
  \institution{Tsinghua University}
  \city{Beijing}
  \country{China}
}
\author{Zhiming Ma}
\email{mazm@amt.ac.cn
}
\affiliation{%
  \department{Academy of Mathematics and Systems Sciences}
  \institution{Chinese Academy of Sciences}
  \city{Beijing}
  \country{China}
}

\renewcommand{\shortauthors}{Trovato et al.}

\begin{abstract}

Learned cardinality estimation requires accurate model designs to capture the local characteristics of probability distributions. However, existing models may fail to accurately capture complex, multilateral dependencies between attributes. Diffusion models, meanwhile, can succeed in estimating image distributions with thousands of dimensions, making them promising candidates, but their heavy weight and high latency prohibit effective implementation. We seek to make diffusion models more lightweight by introducing Accelerated Diffusion Cardest (ADC), the first "downsized" diffusion model framework for efficient, high-precision cardinality estimation. ADC utilizes a hybrid architecture that integrates a Gaussian Mixture-Bayesnet selectivity estimator with a score-based density estimator to perform precise Monte Carlo integration. Addressing the issue of prohibitive inference latencies common in large generative models, we provide theoretical advancements concerning the asymptotic behavior of score functions as time $t$ approaches zero and convergence rate estimates as $t$ increases, enabling the adaptation of score-based diffusion models to the moderate dimensionalities and stringent latency requirements of database systems.

We also introduce ADC+, an optimized variant that dynamically identifies queries with high volume and selectivity, bypassing complex density evaluations in these cases to reduce variance and latency. Through experiments conducted against five learned estimators, including the state-of-the-art Naru, we demonstrate that ADC and ADC+ offer superior robustness when handling datasets with multilateral dependencies, which cannot be effectively summarized using pairwise or triple-wise correlations. In fact, ADC+ is 10 times more accurate than Naru on such datasets. Additionally, ADC+ achieves competitive accuracy comparable to Naru across all tested datasets while maintaining latency half that of Naru's and requiring minimal storage (<350KB) on most datasets.

\end{abstract}


\pagestyle{plain}
\pagenumbering{arabic}
\maketitle

\section{Introduction}
Score-based Diffusion Models have presented the deep learning community with a very, if not the most, attractive way of acquiring highly local information about a complex high-dimensional probability distribution \cite{reftwentyseven}. These models work by training a neural network to approximate the score (gradient of its log function) of an evolving probability density function, which, starting out as the distribution to be learned, is gradually turned into Gaussian noise through a diffusion process. Using the learned score function, score-based diffusion models can then reverse the diffusion process by solving a backward diffusion stochastic differential equation (SDE), and, in doing so, sample from \cite{refthirteen}\cite{reftwelve}\cite{refone} or evaluate \cite{reftwo}\cite{refone} the underlying distribution at the user’s request.

Given their success, one may naturally wonder whether cardinality estimation, a task crucial for the ideal performance of query optimizers, can also benefit from diffusion models. Indeed, several learned cardinality estimation models, such as Naru\cite{refsix} and DQM-D\cite{refseven}, are already performing cardinality estimation by approximating an underlying probability density function at individual points, then numerically integrating it across the queried region using various Monte Carlo schemes\cite{refeleven}\cite{refsix}. However, simply grabbing a popular score-based diffusion model off the shelf probably won’t work due to the three reasons below.

First, while conventional models used for image generation excel at approximating distributions with thousands of dimensions, the sampling process often takes up to tens of seconds, even with the help of advanced TPUs\cite{refeight} and time-saving techniques like consistency models\cite{reffour} or DPM Solver\cite{reften}, which is often longer than the time it costs to actually execute the query.

Second, image generation and cardinality estimation require different magnitudes of precision: while making it five times more likely to generate a specific image gives a generative model its unique art style, predicting 5000 tuples to exist in a region where only 1000 actually do gives a query optimizer an optimization disaster.

Most importantly, existing works on diffusion models have never touched on the problem of integrating the probability density function across a given region, a task useless for image generation, yet crucial for cardinality estimation.

We aim to overcome the aforementioned limitations by presenting Accelerated Diffusion Cardest (ADC) and ADC+ \footnote{See \url{https://github.com/XinheMu/ADC-Replication-} for open source code of ADC}, downsized diffusion models suited for cardinality estimation with their structure described in \textbf{Figure 1}, where
\begin{itemize}
\item ADC answers ranged queries via a prediction-sampling-correction Monte Carlo algorithm, and
\item ADC+ takes a shortcut to skip sampling and correction when answering high volume, high selectivity queries, for which the prediction is already accurate enough.
\end{itemize}
\begin{figure}[h]
  \label{Figure 1}
  \centering
  \includegraphics[width=\linewidth]{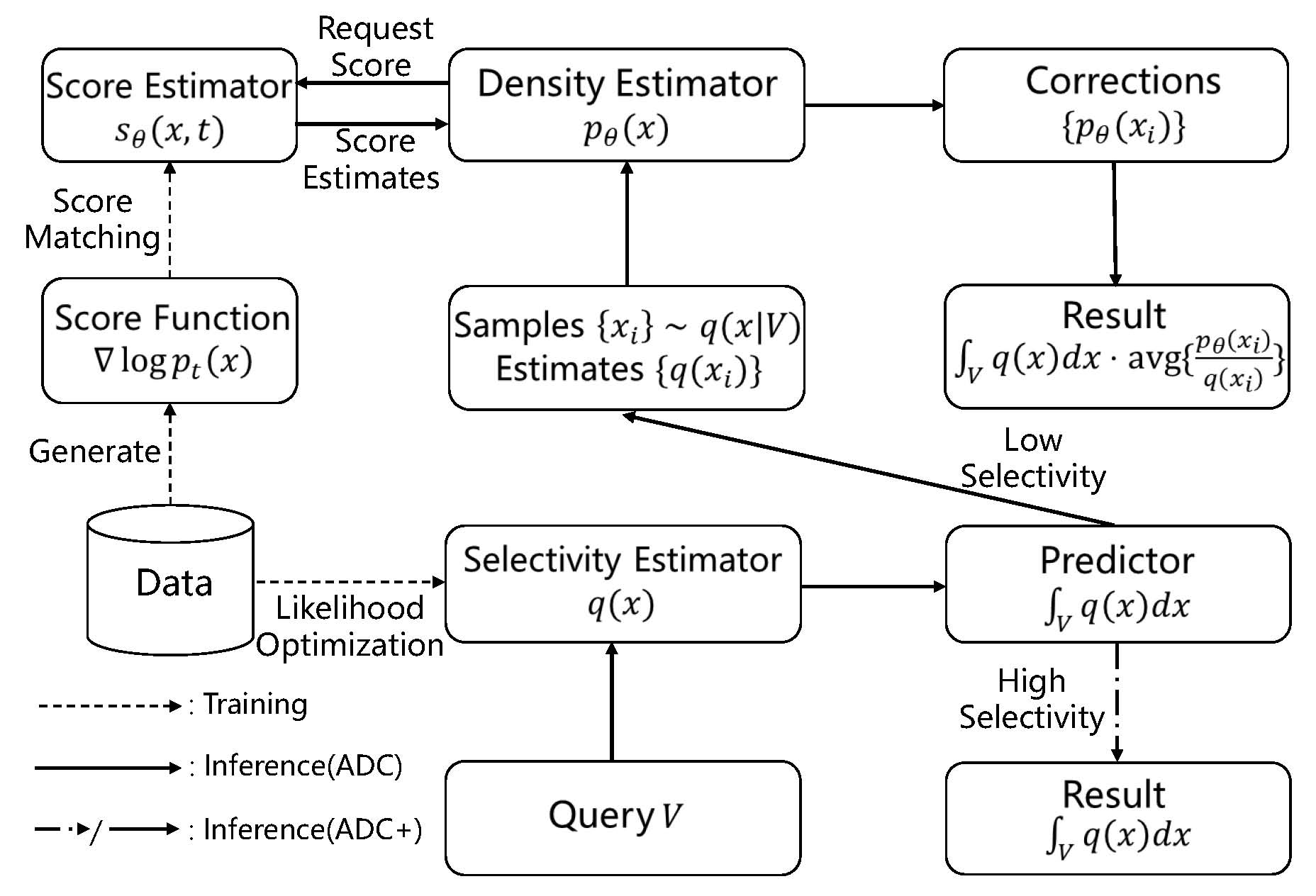}
  \caption{ADC's training and inference framework.}
  \Description{ADC's training and inference workflow. During training, data points are used to (1) directly train the GMM selectivity estimator using maximum likelihood and (2) generate information for the score function, which is then used to train the score estimator. During inference, the GMM selectivity estimator calculates a rough selectivity prediction and conditionally samples points for further evaluation, while the density estimator, accepting input from the score estimator, corrects that prediction by evaluating the probability density function at sampled points with much more accuracy.}
\end{figure}

Here, our major contributions are:
\begin{enumerate}
\item \textit{Downsizing Diffusion Models}. We present the first application of diffusion models to cardinality estimation, providing theoretical guidance that results in more expressive networks and better parameter choices under moderate dimensionalities. We overcome existing challenges by presenting two major findings regarding the diffusion process and the score function. That is, we give asymptotic estimates for the score function as $t\to 0$ around regions with different continuity conditions, helping us mitigate the problem of score blowup by introducing normalization coefficients that grow at different speeds, and we give lower bounds for the speed at which $\nabla\log p_t(x)$ converges to $-\frac{x}{\sigma^2(t)}$ as $t$ grows large, which proves faster than what is generally assumed \cite{reftwentyseven}, helping us choose the integration time-step and $T$, the perturbation upper bound. Using these results, one can design diffusion models that are more accurate and efficient in the database scenario, utilizing their relative freedom from the curse of dimensionality.

\item \textit{Constructing ADC/ADC+}. Using such findings, we design ADC, a diffusion-based algorithm that estimates query cardinality using three modules: a \textit{Score Estimator} $s_\theta(x,t)$ to match the score function $\nabla \log p_t(x)$, a \textit{Density Estimator} $p_\theta(x)$ to estimate pointwise densities using an integration formula for likelihood evaluation \cite{reftwo}, and a Gaussian Mixture Model(GMM)-Bayesnet hybrid \textit{Selectivity Estimator} $q(x)$ to be the predictor in the Monte Carlo formula 
\begin{align}
\label{MC}
\int_V p(x)dx=\int_Vq(x)dx\cdot\mathrm{E}_{q(x|x\in V)}\frac{p(x)}{q(x)}
\end{align}
for ranged queries. ADC+ further uses features "predicted selectivity $\int_Vq(x)dx$" and "query volume" to identify queries for which the variance added by $\mathrm{E}_{q(x|x\in V)}\frac{p(x)}{q(x)}$ exceeds the bias it corrects, in which case it trusts the selectivity estimator to directly output $\int_Vq(x)dx$, resulting in lower median Q-error and latency compared to ADC.
\item \textit{Experimental Evaluation}. We test ADC and ADC+ against five learned estimators (DeepDB\cite{refthirtythree}, Naru\cite{refsix}, LW-NN, LW-Tree\cite{refeighteen}, MSCN\cite{refseventeen}), adhering to well-established benchmarks for dataset choice and workload design\cite{reffive}. We find that ADC and ADC+ are able to learn complex, multilateral dependencies between attributes that cannot be captured by pairwise or triple-wise correlations. In fact, on a synthetic dataset designed to magnify the above trait, ADC+ performs best in the models tested, being 10 times more accurate \footnote{We define "accuracy" as the difference between Q-error and 1, as 1 is Q-error's optimal value} than Naru on 95\%th and 99\%th Q-error. ADC+ also performs well on real-world datasets, rivaling Naru in accuracy and almost always beating all other models, using less than 350KB of storage space and with a latency of less than 10ms per query for the majority of datasets.
\end{enumerate}
\section{Background and Related Work}
\subsection{Cardinality Estimation}
\subsubsection{Problem Formulation} Consider a relation $R$ with attribute domains $\{A_1,...,A_d\}$ and a query $Q$, which can be viewed as a subset $\Omega \subseteq A_1\times A_2\times...\times A_d$, \textit{cardinality estimation} requires us to estimate the query's cardinality, i.e., $|\Omega \cap R|$, without conducting a full search of the relation $R$.

More specifically, existing research on cardinality estimation, such as \cite{reffourteen} and \cite{refeighteen}, often consider an important subproblem in which $\{A_i\}_{i=1}^d$ are all bounded intervals of $\mathbb{R}$ and $\Omega=\prod_{i=1}^dB_i$, with $B_i\subseteq A_i$ a subinterval of $A_i$ (possibly $A_i$ itself). This is also the subproblem we study throughout the paper.

\subsubsection{Current Methods} Current approaches to cardinality estimation can be very roughly categorized into learned and non-learned methods \cite{reffive}. Traditional non-learned methods, such as small-scale sampling, multidimensional histograms \cite{reffifteen}, Bayesnet \cite{refsixteen}, and kernel density estimators \cite{reffourteen}, are usually more time-efficient and relatively robust to constant updates \cite{reffive}, but less accurate than learned models. Thus, they are of more use when the data distribution is simple, when the database is frequently updated, or when the precision requirements are not strict.

Learned models, on the other hand, usually conduct cardinality estimation in one of two ways \cite{reffive}. Regression models, such as MSCN \cite{refseventeen}, LW-tree, and LW-NN \cite{refeighteen}, learn from labeled queries by converting them into feature vectors, then constructing a regression model to match these vectors with the true cardinality of corresponding queries. Besides producing estimation results more accurate than traditional models, these models also have the advantage of taking relatively little time to train, and can make an estimate as quickly as many non-learned methods.

Joint distribution models, meanwhile, assume that data points are distributed according to a probability density function $p(x)$, and thus turn cardinality estimation into a numerical integration problem: that is, calculating $p(x)$ at individual points, then summing them up using a predefined integration scheme. For example, estimators like Naru \cite{refsix} and DQM-D \cite{refseven} calculate the joint distribution function by factorizing it into conditional distributions 
$$p(A_1,...,A_n)=\prod_{i=1}^dp(A_i|A_1,...,A_{i-1}),$$
then sum them up using the Monte Carlo formula (\ref{MC})
with the predictor, $q(x)$, constructed via sequential sampling (Naru), VEGAS \cite{refeleven} (DQM-D), or other algorithms. While joint distribution models cost more time to train and infer, they are known to produce some of the most accurate cardinality estimators.
\subsubsection{Why Diffusion Models?} The high accuracy of joint distribution models invites us to dig deeper into this specific category. However, current joint distribution models all share a subtle but common limitation: conventional methods of constructing the estimator $p$, be it autogression (Naru, DQM-D) or sum-product networks (DeepDB \cite{refthirtythree}), tend to treat a tuple as a collection of correlated, but ultimately distinct, set of attributes rather than a single unified point, risking declines in accuracy if different attributes exhibit complex, multilateral correlations (eg. lying a short distance from several disjoint submanifolds of high curvature) that cannot be efficiently captured by pairwise or even triple-wise correlation alone. While Naru, arguably the state-of-the-art joint distribution estimator, attempts to handle this by factorizing autoregression models in the most expressive way possible (i.e., considering all prior attributes when predicting a new one), it thus experiences a high latency, and, as shown in our experiments, still struggles to capture enough details when the distribution grows too complex.

A diffusion model, on the other hand, generates an image by solving a unified reverse diffusion equation, rather than predicting the values of independent pixel clusters. What's more, their accuracy and expressiveness are repeatedly proven by their success in the generation of images and videos, whose dimensionality often soars in the realm of thousands. Therefore, they bring the potential of overcoming the previous limitation, taking the achievements of joint distribution models one step further.
\subsection{Score-Based Diffusion Models}
\subsubsection{Introduction and Notations} Score-based diffusion models are the generative models that made AI drawing and video generation a reality, thanks to pioneering work by Song, Dickstein and Kingma \cite{refone}, Song and Dhariwal \cite{reffour}, Ho, Jain, and Abbeel \cite{refeight}, Leobacher and Pillichshammer \cite{refthirteen}, and many others. For clarity, we present a summary of notations in \textbf{Table 1}. Each notation will also be explained upon its first appearance.
\begin{table}[h]
    \centering
    \caption{Mathematical Notations}
    \begin{tabular}{c|p{6.5cm}}
        \toprule
        Notation&Definition\\
        \midrule
        $d$&Dimensionality of the dataset\\
        $w$&Standard Brownian motion\\
        $\hat{w}$&Standard Wiener process in reverse time\\
        $\varphi_{\sigma^2}$&Distribution density function of $\mathcal{N}(0,\sigma^2)$\\
        $\tau_k$& Scaling operator $(\tau_k\circ g)(x)=k^dg(kx)$\\
        $\Phi$& Cumulative density function of $\mathcal{N}(0,1)$\\        
        $p_0(x)$&The original data distribution we seek to reconstruct\\
        $k_{\alpha,\beta}(t)$& $e^{\int_0^t\alpha(s)ds}$, abbreviated as $k(t)$\\
        $\sigma^2_{\alpha,\beta}(t)$&Function satisfying $\frac{d}{dt}(\sigma^2(t))=\beta(t)-2\alpha(t)\sigma^2(t)$, $\sigma^2(0)=0$, abbreviated as $\sigma^2(t)$\\        
        $p_{t,\alpha,\beta}(x)$&Distribution derived by perturbing $p_0$ using the diffusion SDE 
        $dx=-\alpha(t)dt+\beta(t)dw$
        from time $0$ to $t$, abbreviated as $p_t(x)$\\
        $p_{0t,\alpha,\beta}(x|x_0)$&Distribution derived by perturbing $\delta_{x_0}$, a distribution concentrated on $x_0$, using the above SDE from time 0 to $t$, abbreviated as $p_{0t}(x|x_0)$.\\    
        $q(x)$&Distribution reconstructed by the GMM predictor.\\
        $q_{t,\alpha,\beta}(x)$&Distribution derived by perturbing $q_0=q$ using the above SDE from time 0 to $t$, abbreviated as $q_{t}(x)$.\\
        $s_\theta(x,t)$&Neural network, often used to approximate $\nabla \log p_t(x)$, with $\theta$ the learnable parameter.\\
        $\epsilon$&Early stopping time implemented upon training $s_\theta(x,t)$, see \textbf{3.3} for why it is needed.\\      $\tilde{s}_\theta(x,t)$&$s_\theta(x,t+\epsilon)$\\
        \bottomrule
    \end{tabular}
    \label{tab:fullwidth}
\end{table}

\subsubsection{Forward and Backward process}Given a set of data points satisfying an unknown distribution $p_0$ in the sample space $\mathbb{R}^d$, score-based diffusion models attempt to learn and sample from $p_0$ in two steps: the \textbf{Forward Process} and the \textbf{Backward Process}\cite{refone}.

The \textbf{Forward Process} gradually injects white noise into the sample points to create Gaussian kernels around them, slowly turning $p_0$ into a smooth Gaussian distribution using the \textit{forward Stochastic Differential Equation (SDE)}
\begin{align}
\label{forward}
dx=-\alpha(t)xdt+\beta(t)dw,t\in [0,T],
\end{align}
where $w$ refers to the standard Brownian motion.

Two most common choices for $\alpha$ and $\beta$ are $\alpha(t)=0,\beta (t)=1$ (the VE scheme) and $\alpha(t)=1,\beta(t)=\sqrt{2}$ (the VP scheme)\cite{refone}. However, works including \cite{refthirtytwo} and \cite{reftwentyfour}  show that we can, in theory, choose any positive function for $\alpha$ and $\beta$ without fundamentally changing the model, as is formally stated by the theorem below. Therefore, we shall assume the VP scheme in all later discussions unless otherwise specified.
\begin{theorem} \footnote{See Appendix A for proof of all theorems}
\label{thm0}
Let $p_0$ be a probabilistic distribution and $\{p_t|t\in[0,T]\}$ be a family of distributions derived by perturbing $p_0$ using the SDE (\ref{forward}) from time 0 to time $T$. Let $k(t)=e^{\int_0^t\alpha(s)ds}$, and $\sigma^2(t)$ be the solution of the initial value Ordinary Differential Equation (ODE)
\begin{align*}
{d\sigma^2}/{dt}&=\beta(t)-2\alpha(t)\sigma^2(t)\\
\sigma^2(0)&=0,
\end{align*}
then, for each choice of $\alpha$ and $\beta$, it would hold that
\begin{align}
p_{t,\alpha,\beta}&=(\tau_{k(t)}\circ p_0)*\varphi_{\sigma^2(t)}\\
\nabla \log p_{t,\alpha,\beta}(x)&=k(t)\nabla \log p_{\sigma^2(t)k^2(t),0,1}(k(t)x),
\end{align}
where $\tau$ in $\mathbb{R}^d$ is the scaling operator $(\tau_k\circ f)(x)=k^df(kx)$ and $\varphi_{\sigma^2(t)}$ in $\mathbb{R}^d$ is the probability density function of $\mathcal{N}(0,\sigma^2(t)I_d)$.
\end{theorem}
The \textbf{Backward Process}, in turn, generates a sample point by first sampling a starting point 
$$\tilde{x}(T)\sim\mathcal{N}(0,\sigma^2(T)I_d)\approx p_{T,\alpha,\beta}(x),$$
then using it as an initial value to solve the \textit{reverse diffusion SDE} \cite{refone}
$$dx=-[\alpha(t)x+\beta^2(t)\nabla \log p_t(x)]dt+\beta(t)d\hat{w}$$
where $\hat{w}$ denotes a standard Wiener process in reverse time.

To approximate the unknown \textit{score function}, $\nabla \log p_t(x)$, diffusion models train a neural network $s_\theta(x,t)$ to minimize
\begin{align}
\label{errorformula}
\int_0^T\mathrm{E}_{p_{t,appr}(x)}[\beta^2(t)\left\|\nabla \log p_{t,appr}(x)-s_\theta (x,t)\right\|_2^2]dt
\end{align}
in a process called \textit{score matching}, with
\begin{align}
p_{t,appr}(x)=\frac{1}{N}\sum_{i=1}^N\mathcal{N}(\frac{x_i}{k(t)},\sigma^2(t)I_d)
\end{align}
being our empirical approximation for $p_t(x)$ from $N$ i.i.d. samples.

\subsubsection{Error Analysis and Likelihood Evaluation} Theorem \ref{estimation}, discovered by Song and Durkan\cite{reftwo}, provides a robust error analysis for diffusion models and a formula for estimating pointwise density values of the underlying distribution.
\begin{theorem}\cite{reftwo}
\label{estimation}
Let $s_\theta(x,t):\mathbb{R}^d\times\mathbb{R^+}\to \mathbb{R}^d$ be a neural network and $p_0$ be a probability distribution, and $p_{0,\theta}(x)$ to be the marginal distribution of $\tilde{x}(0)$ after $\tilde{x}(T)\sim\mathcal{N}(0,\sigma^2(T)I_d)$ is perturbed by the reverse diffusion SDE
$$dx=-[\alpha(t)x+\beta^2(t)s_\theta(x,t)]dt+\beta(t)d\hat{w}$$
from time $t$ to time $0$. Using the notation of Theorem \ref{thm0}, the KL divergence between $p_0$ and $p_\theta$ can be upper bounded by
\begin{align}
\label{KLBound}
D_{KL}\left(p_0\left(x\right),p_{0,\theta}\left(x\right)\right)\leq D_{KL}\left(p_{T,\alpha,\beta},\mathcal{N}\left(0,\sigma^2(T)I_d\right)\right) \notag \\
+\frac{1}{2}\int_0^T\mathrm{E}_{p_{t,\alpha,\beta}(x)}[\beta^2(t)\left\|\nabla \log p_{t,\alpha,\beta}(x)-s_\theta (x,t)\right\|_2^2]dt,
\end{align}
while the value of $\log p_\theta (x_0)$ can be lower bounded by
$$
\log p_{0,\theta}(x_0)\geq\mathrm{E}_{p_{0T}(x|x_0)}\log \varphi_{\sigma^2(T)}(x)+\int_0^T n\alpha(t)dt-\frac{1}{2}\int_0^T\beta^2(t)\cdot 
$$
\begin{align}
\label{density}
\mathrm{E}_{p_{0t(x|x_0)}}&[\left\|\nabla \log p_{0t}(x|x_0)-s_\theta(x,t)\right\|_2^2-\left\|\nabla \log p_{0t}(x|x_0)\right\|_2^2]dt
\end{align}
with the equal sign holding in both inequalities if there exists a probability distribution $q$ such that $s_\theta(x,t)=\nabla\log q_t(x), \forall t \in [0,T]$. 
\end{theorem}

Theorem \ref{estimation} lays the foundation for our model, allowing us to estimate a probability density function at any point, and hence conduct cardinality estimation, by training $s_\theta(x,t)$ and integrating it across time and sample space.

Our model, ADC, is composed of three key modules: the \textit{Score Estimator}, the \textit{Density Estimator}, and the \textit{Selectivity Estimator}. Sections 3 to 5 shall each be devoted to one of the estimators above.

\section{The Score Estimator}
\subsection{Model Choice}
The \textbf {score estimator} trains a neural network, $s_\theta(x,t)$, to approximate the score function $\{\nabla\log p_t(x)|t\in(0,T]\}$, thus providing the density estimator with the data it needs to calculate pointwise densities. Compared to training a network that directly outputs $\nabla\log p_t(x)$, works on score-based diffusion models usually find one of the following two methods much more efficient.

These approaches, named the \textbf{data prediction model} \cite{reftwenty}\cite{refthree} and the \textbf{noise prediction model} \cite{refthree}\cite{refone} respectively, are both inspired by a corollary of (\ref{forward})  that $x(t)\sim p_t(x)$ is the sum of two independent variables
$x_{sig}$ and $x_{noise}$, where $k(t)x_{sig}\sim p_0(x)$ and $\sigma(t)^{-1}x_{noise}\sim\mathcal{N}(0,I_d)$. Therefore, $\nabla\log p_t(x)$ equals
$$\frac{\int_{\mathbb{R}^d}[\tau_{k(t)}\circ p_0](x_{sig})\nabla\varphi_{\sigma^2(t)}(x-x_{sig})dx_{sig}}{\int_{\mathbb{R}^d}[\tau_{k(t)}\circ p_0](x_{sig})\varphi_{\sigma^2(t)}(x-x_{sig})dx_{sig}}=\mathrm{E}_{p_s(x_{sig}|x)}\frac{x_{sig}-x}{\sigma^2(t)},$$
where 
$$p_s(x_{sig}|x)=[\tau_{k(t)}\circ p_0](x_{sig})\varphi_{\sigma^2(t)}(x-x_{sig})$$
denotes the conditional distribution of $x_{sig}$ for a fixed $x$.

As a result, the \textbf{data prediction model} works by having $v_\theta(x,t)$ learn to match
$\mathrm{E}_{p_s(x_{sig}|x)}k(t)x_{sig}$, the normalized conditional expectation of $x_{sig}$ for a given $x$, and then output $$\nabla\log p_t(x)\approx s_\theta(x,t)=\frac{v_\theta(x,t)}{\sigma^2(t)k(t)}-\frac{x}{\sigma^2(t)}.$$
Similarly, defining 
$$p_n(x_{noise}|x)=[\tau_{k(t)}\circ p_0](x-x_{noise})\varphi_{\sigma^2(t)}(x_{noise}),$$
the \textbf{noise prediction model} has $v_\theta(x,t)$ learn to match $$\mathrm{E}_{p_n(x_{noise}|x)}\frac{x_{noise}}{\sigma(t)},$$
and then output $\nabla \log p_t(x)\approx s_\theta(x,t)=\frac{v_\theta(x,t)}{\sigma(t)}$.

Conventional works on diffusion models usually stick to one of the two designs throughout the score matching process. However, we can easily predict which design works better at a given time $t$. Indeed, as $x=x_{sig}+x_{noise}$, \textit{absolute error} from approximating $\mathrm{E}_{p_s(x_{sig}|x)}x_{sig}$ and $\mathrm{E}_{p_n(x_{sig}|x)}x_{noise}$ equally affect the error term $\left\|\nabla \log p_{t,approx}(x)-s_\theta (x,t)\right\|_2^2$ for any choice of $(x,t)$. Thus, predicting the \textit{smaller} of $x_{sig}$ and $x_{noise}$ would reduce the denominator for \textit{relative error}, thereby leading to a larger \textit{relative error} tolerance.

 A common property of the forward diffusion process is that the signal-to-noise ratio, ${\mathrm{E}(x_{sig})}/{\mathrm{E}(x_{noise})}$, decreases uniformly with $t$. Thus, we have ADC's score estimator break score matching into two parts: approximating $\nabla \log p_t(x)$ with a data prediction model $s_{\theta_2,tail}(x,t)$ when $t$ is above a threshold $T_{trunc}$, and a noise prediction model $s_{\theta_1,head}(x,t)$ when $t$ is below it. To the best of our knowledge, this is the first time such a practice has been adopted in the field of score-based diffusion models.

Many studies \cite{reftwentyseven} \cite{refnineteen} \cite{reffour} show that $p_{t,appr}(x)$, when constructed from i.i.d. samples, grows drastically larger and more volatile as $t\to 0$ and eventually explodes to infinity. As a result, $s_{\theta_1,head}$ must be somewhat heavy to produce moderately accurate estimations, but $s_{\theta_2,tail}$ can be very lightweight yet just as precise, especially since when $p_t$ gets sufficiently close to $\mathcal{N}(0,\sigma^2(t))$, the data prediction model's added term, 
$$-\frac{x}{\sigma^2(t)}=\nabla\log \varphi_{\sigma^2(t)}(x),$$
has already completed the bulk of score matching pretty decently. Therefore, our approach of implementing two separate models not only improves precision but can also speed up computation by calculating the score function with a heavy network only where it is really needed. Finally, the extra storage space cost incurred by training two separate models is quite minimal: On all benchmark datasets, the extra network used by the data prediction model, $s_{\theta_2,tail}$, only takes up roughly a third of $s_{\theta_1,head}$'s storage space.
\subsection{Model Architecture Improvement: QuadNet}
Due to difficulties encountered in training, we modify the noise prediction model based on the observation that the term $x_{noise}$ might or might not cancel itself out when calculating $\mathrm{E}_{p_n(x_{noise}|x)}\frac{x_{noise}}{\sigma(t)}$, making the growth pattern of $\nabla \log p_t(x)$ vary significantly depending on the smoothness of $p_0(x)$ instead of being fixed to a uniform $\mathcal{O}(\frac{1}{\sigma(t)})$, as is assumed by the traditional model design.

The theoretical foundation of our proposed modifications is summed up in the theorem below.
\begin{theorem}
\label{thm1}
Let $p_0$ be the weighted average of several positive distributions $\{p_i\}_{i=1}^n$, such that
$$\{V_i\}_{i=1}^n:=\{x|p_i(x)>0\}_{i=1}^n$$
are open domains in $\mathbb{R}^d$ intersecting only at their boundaries, and each $p_i$ is Lipschitz with constant $k_i$ in the interior (but not necessarily at the boundary) of $V_i$. Denote $V_0=\mathbb{R}^d\textbackslash\cup_{i=1}^n V_i$
for convenience, and let $\{p_t|t\in(0,T]\}$ be the family of distributions derived by perturbing $p_0$ according to the VE\cite{refone} scheme (results for the VP\cite{refone} scheme can be immediately derived from Theorem \ref{thm0}), then:
\begin{enumerate}
\item $\forall x\in V_i$ such that $B(x,r)\subseteq V_i$, 
$$\nabla \log p_t(x_0)=\nabla \log p_0(x_0)+\delta_1(x_0,t),$$
with $\delta_1(x_0,t)$ an infinitesimal term of order $O(\sigma(t))$.
\item $\forall x \in \partial V_i\cap\partial V_j$ with at most one of $i$ and $j$ being zero, if $B(x_0,r)\subseteq V_i\cup V_j$, $p_{0,i}(x_0)+p_{0,j}(x_0)>0$, and inside $B(x_0,r)$, 
$$\mathrm{d}\left(y,T_{x_0}(\partial V_i\cap \partial V_j)\right)\leq h\left\|y-x_0\right\|_2^2$$ holds for some constant $h$ and all $y\in\partial V_i\cap\partial V_j$, then $\forall \lambda$, $\nabla \log p_t(x_0+\lambda\sqrt{t}\vec{n})$ would equal
$$
\frac{e^{-\frac{\lambda^2}{2}}(p_{0,i}(x_0)-p_{0,j}(x_0))\vec{n}}{2\sigma(t)\left(\Phi(\lambda)p_{0,i}(x_0)+(1-\Phi(\lambda)p_{0,j}(x_0))\right)}+\delta_\sigma(x_0,\lambda,t),
$$
with $\vec{n}$ the unit vector normal to $T_{x_0}(\partial V_i\cap\partial V_j)$, pointing from $V_j$ to $V_i$, and $\delta_\sigma(x_0,\lambda,t)$ a residual term of order $O(1)$.
\item $\forall x\in V_0$, suppose $y_0\in \bar{V_i}$ is the point closest to $x_0$ such that $y_0\in\cup_{i=1}^k\bar{V_i}$. If $P(x\in B_r(y)|x\sim p_0)>M_1r^k, \forall 0<r<R$ for some constants $M_1, k_1$, while
$$\left\|z-x_0\right\|_2^2-\left\|y_0-x_0\right\|_2^2\geq \min\{M_2,k_2\left\|z-y_0\right\|_2^2\}$$
for all $z\in\cup_{i=1}^kV_i$ and some constants $M_2,k_2$, then
$$\nabla\log p_t(x_0)=\frac{y_0-x_0}{\sigma^2(t)}+\delta_{\sigma^2}(x_0,t),$$
with $\delta_{\sigma^2}(x_0,t)$ a term of order $o(\frac{1}{\sigma(t)^{1+\epsilon}})$ for all $\epsilon>0$.
\end{enumerate}
\end{theorem}
Such estimates inspire us to present Quadnet, a new network structure. That is, for a $d$-dimensional dataset, we have our network $s_{\theta_1,head}$ generate a vector $(v_{\sigma^2},v_\sigma,v_1)$ of length $3d$, and output
$$\nabla\log p_t(x)\approx \frac{v_{\sigma^2}}{\sigma^2(t)}+\frac{v_\sigma}{\sigma(t)}+v_1,$$
so that each module effectively captures the behavior of $\nabla \log p_t(x)$ near one type of region, while correcting the residual error term near the more "explosive" regions.

To illustrate the advantage of Quadnet, we define a toy dataset as of \textbf{Figure 2}, and have a Quadnet of shape $[3,10,25,25,10,6]$ compete against three compare group networks of shape $[3,10,25,25,10,$ $6,2]$ in performing score matching for small $t$. Compared to its competitors that possess only the quadratic (${1}/{\sigma^2(t)}$), linear (${1}/{\sigma(t)}$, i.e. a traditional noise prediction model), or constant (i.e. one that directly approximates $\nabla \log p_t(x)$) scaling module, Quadnet converges faster, settles to a lower error, and becomes less prone to falling into saddle points than all of its competitors. 
\begin{figure}[h]
  \label{Figure 2}
  \centering
  \includegraphics[width=\linewidth]{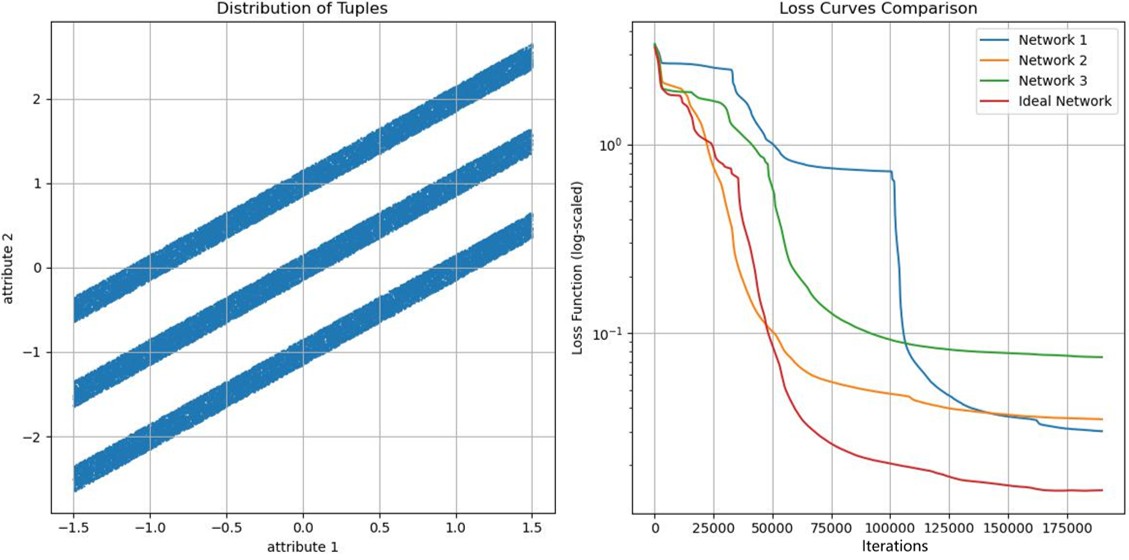}
  \caption{The 2d toy dataset for testing (left) and the loss curve of different score matching network architectures (right). "Ideal network" (red) refers to Quadnet while Network 1 (blue), 2 (yellow), and 3 (green) possess only the quadratic, linear, and constant scaling module, respectively.}
  \Description{Left: a scatterplot showing all data points in the test dataset. Right: the loss curve of Quadnet and three compare group networks when conducting score matching on the test dataset, with x-axis depicting the number of iterations and y-axis depicting the current loss function (log scaled.)}
\end{figure}
\subsection{Choosing Early Stopping Time}
Despite the enhancements above, it is widely observed that the score function would still grow to a point beyond the approximation ability of \textit{any} neural networks when $t$ becomes too close to zero. Score-based diffusion models used for image generation mitigate this problem by implementing early stopping \cite{refnineteen}\cite{reffour}: that is,  training the score matching model (i.e., calculating the score approximation loss) and running the reverse diffusion SDE solver on the interval $(\epsilon,T]$ rather than $(0,T]$, and treating $x(\epsilon)$ rather than $x(0)$ as the final sample.

The score estimator in ADC is trained using the same early stopping approach, with $\epsilon$, the early stopping time, selected among values $\{\frac{1}{1280},\frac{1}{640},\frac{1}{320},\frac{1}{160}\}$ using a trial-and-error approach (starting from $\frac{1}{1280}$ and increasing $\epsilon$ if the MSE loss proves too big). We notice that the best choice for $\epsilon$ increases with dimensionality and volume of the data distribution's typical set, but more work is still needed to determine the best choice for $\epsilon$ ahead of model construction.

\section{The Density Estimator}
\subsection{Model Design}
Relying on $s_\theta(x,t)\approx\nabla\log p_t(x)$, as input, the density estimator calculates pointwise joint densities using the formula (\ref{density}), in which we denote 
$$\left\|\nabla \log p_{0t}(x|x_0)-s_\theta(x,t)\right\|_2^2-\left\|\nabla \log p_{0t}(x|x_0)\right\|_2^2$$
as $g_{x_0}(x,t,s_\theta)$ for convenience. In the formula (\ref{density}), 
$$f(x_0,\alpha,\beta,T)=\mathrm{E}_{p_{0T}(x|x_0)}\log \varphi_{\sigma^2(T)}(x)+\int_0^T n\alpha(t)dt$$ 
can be calculated analytically, thus leaving the estimation of
$$\frac{1}{2}\int_0^T\beta^2(t)\mathrm{E}_{p_{0t(x|x_0)}}g_{x_0}(x,t,s_\theta)dt$$
to be the bulk of our work.

Again, due to early stopping and score blowup, our network $s_\theta(x,t)$ is not (and cannot be) trained to produce the value of $\nabla \log p_\delta(x)$ for all $\delta<\epsilon$, potentially leading to wildly inaccurate representations. There are three ways to deal with this: we can evaluate the integral while skipping the $[0,\epsilon)$ interval, we can assume $s_\theta(x,\delta)=s_\theta(x,\epsilon)$ for all $\delta<\epsilon$, or we can define $\tilde{s}_\theta(x,t)=s_\theta(x,t+\epsilon)$, $\tilde{\alpha}(t)=\alpha(t+\epsilon)$, $\tilde{\beta}=\beta(t+\epsilon)$ and calculate 
$$\log p_{\epsilon,\theta}(x_0)\approx f(x_0,\tilde{\alpha},\tilde{\beta},T)+\int_0^T\tilde{\beta}^2(t)\mathrm{E}_{p_{0t(x|x_0)}}g_{x_0}(x,t,\tilde{s}_\theta)dt,$$
which would be an approximation of $$p_\epsilon(x)=[(\tau_{k(\epsilon)}\circ p_0)*\varphi_{\sigma^2(t)}](x)$$
more numerically stable.

ADC chooses the third approach despite the added perturbation, both because it's easier to analyze theoretically, and because the first two approaches seem to introduce arbitrary spikes in the approximated joint distribution function that significantly reduce the accuracy of the current selectivity estimator, designed using a predictor-corrector importance sampling algorithm. Whether the first two approaches might work for a different selectivity estimator design can be a topic of future research.

The above formula requires us to integrate $g_{x_0}$ across both sample space (when calculating $\mathrm{E}_{p_{0t(x|x_0)}}g_{x_0}(x,t,\tilde{s}_\theta)$) and time. Here, time comes at a much lower dimension than space. What's more, we generally know much more about how the norm of $\nabla \log p_t(x)$ changes with respect to time than how it does with respect to space, as the latter greatly depends on the underlying distribution. 

To account for such, we present a hybrid integration approach: for integration across time, ADC implements an adaptive stepsize Midpoint Rule Integrator, which is more suited for low-dimensional functions and can be easily modified to capitalize our knowledge of the score function's behavior; for calculating the expectation across space, ADC implements a Quasi Monte Carlo \cite{reftwentyone} integrator much more robust against the curse of dimensionality.\\
Summing it up, for our chosen VP perturbation scheme where $\alpha(t)=1$ and $\beta(t)=\sqrt{2}$ are set to constants, ADC approximates 
$$\int_0^T\mathrm{E}_{p_{0t}(x|x_0)}g_{x_0}(x,t,\tilde{s}_\theta)dt$$
using
\begin{align}
\label{QMC}
\sum_{i=1}^N\frac{(t_i-t_{i-1})}{2^k}\Big[\sum_{j=1}^{2^k}g_{x_0}\Big(\sigma(\tilde{t}_i)\Phi^{-1}\big(z_j+y(\tilde{t}_i)\big)+\frac{x_0}{k(\tilde{t}_i)},\tilde{t_i},\tilde{s}_\theta\Big)\Big],
\end{align}
where $0=t_0<t_1<...<t_{N-1}<t_N=T$ is a partition of $[0,T]$, $\tilde{t}_i={(t_{i-1}+t_i)}/{2}$; $\{z_j\}_{j=1}^{2^k}$ is a Sobol sequence of length $2^k$ in the unit cube; $\{y(\tilde{t}_i)\}_{i=1}^N$ are $N$ random samples from $\mathrm{U}([0,1]^d)$, used to perturb $\{z_j\}$ so that the integrator will not focus too much on any specific region; $\Phi$ is the distribution function of the standard normal distribution, acting pointwise on each coordinate; and the sum between $z_j$ and $y(\tilde{t}_i)$ is taken in $\mathrm{T}^n$ rather than $\mathbb{R}^n$ to prevent out-of-domain inputs for $\Phi^{-1}$.
\subsection{Integrand Choice}
Based on the data prediction model, we further propose modifying the expression of $g_{x_0}(x,t,\tilde{s}_\theta)$ for $t>T_{trunc}$ in a way that drastically reduces integration error and allows for the use of much larger time-steps.

Our idea is based on the following observations. First, for large $t$, $s_\theta(x,t)$ would quickly approach $-\frac{x}{\sigma^2(t)}$. Second, the data prediction model is already using $\frac{v_\theta(x,t)}{\sigma^2(t)k(t)}$ to approximate the residual term $\nabla \log p_t(x)+\frac{x}{\sigma^2(t)}$ rather than $\nabla \log p_t(x)$ itself. Most importantly,  
$$\mathrm{E}_{p_{0t}(x|x_0)}\big[\left\|\nabla \log p_{0t}(x|x_0)+\big(x/\sigma^2(t)\big)\right\|_2^2-\left\|\nabla \log p_{0t}(x|x_0)\right\|_2^2\big]$$
is exactly calculated as $$\frac{\beta^2(t)(n\sigma^2(t)k^2(t)+\left\|x_0\right\|_2^2)}{\sigma^4(t)k^2(t)},$$ which can be integrated analytically on any interval whenever $\alpha(t)$ and $\beta(t)$ are constants.

With the above in mind, for all $t>T_{trunc}$, ADC decomposes $g_{x_0}(x,t,\tilde{s}_\theta)$ into the terms
$$\frac{1}{\sigma^4(t)k^2(t)}\langle\tilde{v}_{\theta}(x,t),\tilde{v}_{\theta}(x,t)-2x_0\rangle-\langle\frac{x}{\sigma^2(t)},\frac{x}{\sigma^2(t)}-\frac{2x_0}{\sigma^2(t)k(t)}\rangle,$$
leaves the second term for exact calculation, and only approximates the first term, denoted from now on as $g_{res,x_0}(x,t,\tilde{v}_\theta)$, using hybrid numerical integration.

This modification is meaningful in that as $t$ gets larger, $g_{x_0}(x,t,\tilde{s}_\theta)$, the original integrand, would approach a nonzero value that varies notably with $x$, but $g_{res,x_0}(x,t,\tilde{v}_\theta)$ converges to zero at a rate no slower than $\mathcal{O}(\frac{1}{\sigma^4(t)k^2(t)})$, i.e., $\mathcal{O}(e^{-2t})$ under the VP scheme. As a result, ADC is allowed to implement much larger time steps for even moderately large $t$, significantly boosting its efficiency.
\begin{algorithm}[H]
\caption{Set Time-step Scheme for Density Estimator}
\label{Algorithm1}
\begin{algorithmic}[1] 
\Require Time threshold $T_{trunc}$, early stopping time $\epsilon$, baseline step-size $\delta_{head}$, $\delta_{tail}$ 
\Ensure Time-step scheme $\{t_i\}_{i=0}^N$
\State Set $s_0\gets T_{trunc}$, $i\gets 0$
\While {$s_i>0$} 
    \State $s_{i-1}\gets t_i-\delta_{head}\cdot\frac{\sigma(s_i)\sigma(s_i+\epsilon)}{\sigma(T_{trunc})\sigma(T_{trunc}+\epsilon)}$, $i\gets i-1$
\EndWhile
\State $s_i\gets 0$, $N_1 \gets i$
\State $i\gets 0$
\While {$s_i<T$} 
    \State $s_{i+1}\gets s_i+\delta_{tail}\cdot \frac{\sigma^4(T_{trunc})k^2(T_{trunc})}{\sigma^4(s_i)k^2(s_i)}$, $i\gets i+1$
\EndWhile
\State $s_i \gets T$, $N\gets i-N_1$
\State $\{t_i\}_{i=0}^N \gets \{s_{i-N_1}\}_{i=0}^N$
\State \Return $\{t_i\}_{i=0}^N$
\end{algorithmic}
\end{algorithm}
\subsection{Parameter Choice}
We discuss our rationale for choosing two parameters: $T$, the integration upper bound, currently set to 3; and $\{t_i\}_{i=0}^N$, the timestep scheme, currently chosen via the pseudocode in \textbf{Algorithm 1}.

To choose $T$, we consider the following theorem that provides a lower bound for the speed at which $p_t$ converges to a known Gaussian distribution $\mathcal{N}(0,\sigma^2(t)I_d)$.
\begin{theorem}
Let $\tilde{x}$ be a random variable taking values $\{x_i\}_{1\leq i\leq n}$ with equal probability (as is \textbf{always} the case when we construct $p_t$ from i.i.d. sample points), such that $\frac{1}{n}\sum_{i=1}^nx_i=0$. Take $p_t(x)$ to be the marginal distribution of $x(t)$ when $x(0)=\tilde{x}$ is perturbed by the diffusion SDE
$$dx=-\alpha(t)xdt+\beta(t)dw$$
from time 0 to time $t$. Then, we would have 
\begin{enumerate}
\item
$\mathrm{E}_{p_t(x)}\left\|\nabla\log p_t(x)+\frac{x}{\sigma^2(t)}\right\|_2^2\leq \frac{f(t,\tilde{x})}{\sigma^4(t)k^2(t)}.$
\item
$\mathrm{E}_{p_t(x)}\left\|\nabla \log p_t(x)+\frac{x}{\sigma^2(t)}\right\|_2^2\leq \frac{f(t,\tilde{x})g(t,\tilde{x})}{\sigma^4(t)k^2(t)}.$
\end{enumerate}
where
$$f(t,\tilde{x})=\frac{1}{n}\sum_{i=1}^n\left\|x_i\right\|_2^2$$
and
$$g(t,\tilde{x})=\frac{1}{n}\sum_{i=1}^n(e^{\frac{2\left\|x_i\right\|_2^2+f(t,\tilde{x})}{2\sigma^2(t)k^2(t)}}-1).$$
\end{theorem}
The above theorem tells us that under the condition $\mathrm{E}_{p_0(x)}x=0$, predicting $\nabla \log p_t(x)$ using the score function of a normal distribution results in an error term that decays at a speed of roughly $\mathcal{O}(\frac{1}{\sigma^4(t)k^2(t)})$ for small $t$, and roughly $\mathcal{O}(\frac{1}{\sigma^6(t)k^4(t)})$ for large $t$, with the transition guaranteed to happen when 
$$\frac{2\left\|x_i\right\|_2^2+\mathrm{E}_{p_0(x)}\left\|x\right\|_2^2}{2\sigma^2(t)k(t)}\leq \log 2$$
holds for all $\{x_i\}_{i=1}^n$. Meanwhile, one can deduce from (\ref{KLBound}) that
$$D_{KL}(p_T(x),\mathcal{N}(0,\sigma^2(T)))=\int_T^\infty\mathrm{E}_{p_t(x)}\left\|\nabla\log p_t(x)+\frac{x}{\sigma^2(t)}\right\|_2^2dt,$$
so setting $T$ slightly beyond this transition point might be the most cost-effective approach.

In the case of ADC, our normalization process ensures $\mathrm{E}_{p_0(x)}=0$, and each attribute's $\max$ and $\min$ value differ by 3.2. As a result, under the VP scheme, setting $T=3$ would suffice for all datasets of moderate dimensionality.

To choose the timestep scheme, we first perform an error analysis on our hybrid integrator. For an arbitrary timestep partition scheme $\{t_i\}_{i=0}^N$, define the time-related discretization error $\delta_{i,time}(x_0)$ as
$$\int_{t_i}^{t_{i+1}}[\mathrm{E}_{p_{0t}(x|x_0)}g_{x_0}]dt-\Delta t_i\mathrm{E}_{p_{0\tilde{t}}(x|x_0)}g_{x_0}(x,\tilde{t}_i,\tilde{s}_\theta),$$
and the space-related discretization error $\delta_{t,space}(x_0,y)$ as
$$\mathrm{E}_{p_{0t}(x|x_0)}g_{x_0}(x,t,\tilde{s}_\theta)-\frac{1}{2^k}\sum_{j=1}^{2^k}g_{x_0}\big(\sigma(t)\Phi^{-1}\big(z_j+y\big)+\frac{x_0}{k(t)},t,\tilde{s}_\theta\big).$$
Given the curse of dimensionality and that the volatile $g(x,t,s_\theta)$ was already "smoothed up" somewhat upon calculating $\mathrm{E}_{p_{0t}(x|x_0)}g_{x_0}$, one can naturally expect $\{\delta_{i,time}(x_0)\}_{i=1}^N$ to be significantly smaller than $\{\delta_{\tilde{t}_i,space}(x_0,y(\tilde{t_i}))\}_{i=1}^N$ with a dense enough timestep scheme, in which case we have the following theorem.
\begin{theorem}
Define $\mathrm{Var}(x_0)$ to be the error term's variance upon approximating $\log p_{\epsilon,\theta}(x_0)$ with formula (\ref{QMC}). Assuming that 
\begin{enumerate}
\item The partition is dense enough that for the same $y$, fluctuations of $\Delta t_i\delta_{t,space}(x_0,y)$ for different $t$ in $[t_{i-1},t_{i}]$ is negligible compared to the total integration error,
\item The error terms $\{\delta_{i,time}(x_0)\}_{i=1}^N$ are negligible compared to $\{\delta_{\tilde{t}_i,space}(x_0,y(\tilde{t_i}))\}_{i=1}^N$, and
\item The sequence $\{y(\tilde{t}_i)\}_{i=1}^N$ consists of i.i.d. samples from $\mathrm{U}([0,1]^d)$
\end{enumerate}
hold for all $x_0$, $\mathrm{E}_{x_0\sim p_0}\mathrm{Var}(x_0)$ would be minimized under a timestep scheme where $t_i-t_{i-1}$ is chosen inversely proportionate to $$\big(\mathrm{E}_{x_0\sim p_0}[\mathrm{Var}_{y}\delta_{\tilde{t}_i,space}(x_0,y)]\big)^{\frac{1}{2}}.$$
\end{theorem}
One issue yet unresolved is that $\mathrm{E}_{x_0\sim p_0}[\mathrm{Var}_{y}\delta_{\tilde{t}_i,space}(x_0,y)]$ proves very hard to estimate. As a result, the current version of ADC (somewhat boldly) deems the term proportionate to $$\mathrm{E_{x_0\sim p_0}\mathrm{E}_{p_{0t(x|x_0)}}}\big(g_{x_0}(x,t,\tilde{s}_\theta)\big)^2$$
when $t<T_{trunc}$ and 
$$\mathrm{E_{x_0\sim p_0}\mathrm{E}_{p_{0t(x|x_0)}}}\big(g_{res,x_0}(x,t,\tilde{s}_\theta)\big)^2$$
when $t>T_{trunc}$, leaving the task of more accurately estimating  $\mathrm{E}_{x_0\sim p_0}[\mathrm{Var}_{y}\delta_{\tilde{t}_i,space}(x_0,y)]$ for upcoming research.

The decay rate of the second term is estimated to be no slower than ${O}({1}/(\sigma^8(t)k^4(t)))$ when $t\to\infty$, as we already discussed. For the first term, Theorem \ref{thm1} proves that for all $x_0$ such that $p_0(x_0)\neq 0$, $\mathrm{E}_{p_{0t}(x|x_0)}\left\|\nabla\log p_{t+\epsilon}(x)\right\|_2^4$ can be upper bounded by
$$ 8\mathrm{E}_{p_{0t}(x|x_0)}\left\|\frac{x-x_0}{\sigma^2(t+\epsilon)}\right\|_2^4+8\mathrm{E}_{p_{0t}(x|x_0)}\left\|\delta_{\sigma^2}(x,t+\epsilon)\right\|_2^4,$$
which, using the bounds on $\delta_{\sigma^2}(x,t)$, can be roughly estimated as proportionate to $\frac{1}{\sigma^4(t+\epsilon)}$. Therefore, Cauchy's Inequality proves
$$\mathrm{E}_{p_{0t}(x|x_0)}(\langle\nabla\log p_{t+\epsilon}(x),\nabla\log p_{t+\epsilon}(x)-2\nabla\log p_{0t}(x|x_0)\rangle)^2$$
to grow at a rate no faster than ${O}({1}/({\sigma^2(t)\sigma^2(t+\epsilon)}))$ as $t\to 0$, thus justifying our pseudocode implementation.
\section{The Selectivity Estimator}
\subsection{Predictor Choice: Why GMM}
While the score estimator and density estimator are effective in calculating the probability density function's value at any given point, we still need to integrate this function over the entire queried region in order to estimate the selectivity of ranged queries. ADC does this by utilizing the importance sampling Monte Carlo formula (\ref{MC}) that Naru \cite{refsix} and DQM-D\cite{refseven}'s selectivity estimator are both based on, which, when $\mathrm{E}_{q(x|x\in V)}\frac{p(x)}{q(x)}$ is approximated using i.i.d. sampling, becomes an unbiased estimator with variance proportionate to 
$$\Big(\int_Vq(x)dx\Big)^2\mathrm{Var}_{x\sim g(x| V)}\Big(\frac{p(x)}{q(x)}\Big).$$
To enhance speed and precision, the selectivity estimator of ADC must construct a predictor that can be rapidly integrated, rapidly sampled from, and resembles the corrector $p_{\epsilon,\theta}$ as much as possible. However, the predictors constructed by Naru (via progressive sampling) and DQM-D (via VEGAS \cite{refeleven}) require multiple phases of sampling. This is undesirable for us, considering that evaluating the density function at one point already requires a not-so-easy numerical integration process.

Gaussian Mixture Models (GMMs) attempt to reconstruct a distribution $q$ using the expression
\begin{align}
\label{GMM}
q=\sum_{i=1}^Nw_i\mathcal{N}(y_i,H_i),
\end{align}
and ADC's choice to train one as its predictor is born out of the very same considerations listed above. 

In terms of speed, the advantage of GMM manifests in that for $\{H_i\}_{i=1}^N$ diagonal and $V=\{(x_1,..,x_d)|a_j\leq x_j\leq b_j, \forall j\}$ a hyper-rectangle in $\mathbb{R}^d$, $\int_Vq(x)dx$ can be calculated as 
$$\int_Vq(x)dx=\sum_{i=1}^Nw_i\prod_{j=1}^d\left[\Phi(h_{jj}(b_j-y_j))-\Phi(h_{jj}(y_j-a_j))\right],$$
and conditional samples $\{x_j\sim q(x|V)\}$ can be drawn using $d+1$ random numbers $z_i\sim U[0,1]$ per sample, with $z_{n+1}$ determining which kernel $\mathcal{N}(y_j,H_j)$ $x$ belongs to, and $z_i$ determining the $i$th coordinate of $x$ to be 
$$x_i=\Phi^{-1}\big(z_i\Phi(h_{ii}(b_i-y_{j,i}))+(1-z_i)\Phi(h_{ii}(y_{j,i}-a_i))\big).$$
As a result, when $V$ is a hyperrectangle, sampling and integration can both be performed with a time cost independent of $V$'s diameter and comparable to that of calculating a pointwise density value.

In terms of similarity to $p$, another advantage of GMM is that the underlying distribution $p_\epsilon$, which $p_{\epsilon,\theta}$ seeks to approximate, always has a KDE\cite{reffourteen}-type expression
$$p_\epsilon (x)=\frac{1}{M}\sum_{i=1}^M\mathcal{N}(k(\epsilon)^{-1}(x_i),\sigma^2(\epsilon)),$$
as in the database case, $p_0$ is simply the average of delta functions placed at sample points. Given their unique shape and smoothness features, we can reasonably assume that Gaussian mixtures work best at matching local features of Gaussian mixtures.

Instead of a query-based training approach, ADC trains its GMM predictor to maximize the log-likelihood function
\begin{align}
\label{Likelihood}
\frac{1}{M}\sum_{j=1}^M\log q(x_j).
\end{align}
using the EM algorithm \cite{refthirtyseven}, an elegant GMM-exclusive optimizer that trains GMMs by iteratively alternating between estimating the posterior responsibilities of mixture components (E-step) and updating the model parameters to maximize the expected complete-data log-likelihood (M-step). Compared to one-size-fits-all optimizers like Stochastic Gradient Descent (SGD), EM is very robust and efficient, as it requires no hyperparameters, guarantees a monotone increase of likelihood, and always converges to a local maximum value.

An important detail of our GMM implementation is based on the following observations: first, $q$, when trained to maximize the function (\ref{Likelihood}), is in fact trying to approximate $p_0=\frac{1}{M}\sum_{i=1}^M\delta(x_i)$ itself rather than the perturbed $p_\epsilon$. Second, perturbing GMMs using a diffusion process is extremely easy: setting $q_0(x)$ to be the $q(x)$ defined in (\ref{GMM}), we would simply have 
$$q_t=\sum_{i=1}^Nw_i\mathcal{N}\Big(\frac{y_i}{k(t)},\frac{H_i}{k(t)}+\sigma^2(t)I_d\Big).$$
With this in mind, ADC's selectivity estimator assumes that 
$$\frac{\int_Vp_{0,\theta}(x)dx}{\int_Vp_{\epsilon,\theta}(x)dx}\approx\frac{\int_Vq_0(x)dx}{\int_Vq_\epsilon(x)dx},$$
and then corrects for the early stopping error by actually calculating
$$\int_Vp_{0,\theta}(x)dx\approx\int_Vq_0(x)dx\cdot\mathrm{E}_{x\sim q_\epsilon(x|V)}\Big(\frac{p_{\epsilon,\theta}(x)}{q_\epsilon(x)}\Big).$$
While a somewhat bold approach, numerical experiments have proven such a design crucial for ADC's optimal performance. 
\subsection{A GMM-Bayesnet Hybrid}
A common flaw shared by ADC's predictor and corrector is that GMMs and diffusion models both work by "bleeding" portions of density from the typical set onto neighboring regions. When one or more pairs of attributes $(A_i,A_j)$ exhibit near-functional dependency, hence the typical set's projection onto $(A_i,A_j)$ becomes a 1-dimensional curve, such a practice significantly distorts the local probability distribution, leading to large errors for extreme-case queries.

Thus, ADC further introduces Bayesnet \cite{refsixteen} to handle functionally dependent attribute pairs. In the case where knowing the value of an attribute $A_i$ can help narrow down the range of another attribute $A_j$ to $c$ (chosen to be 0.05 in the current model) times the original, ADC, instead of training a GMM and a diffusion model on all attributes, would instead work by training:
\begin{itemize}
\item a diffusion model $p_\epsilon$ and a GMM $q$ modeling the distribution of all attributes excluding $A_j$
\item a Bayesnet histogram modeling the conditional probability density $\tilde{p}(A_j|A_i)$ as a function piecewise constant with respect to $A_i$ and $A_j$.
\end{itemize}
Now, given a query asking for all points inside $V=\prod_{k=1}^d[a_k,b_k]$, ADC predicts its selectivity using 
$$\sum_{m=1}^Nw_m\left(\prod_{k\neq{i,j}}\int_{a_k}^{b_k}\varphi_{H_{m,kk}}(x_k-y_{m,k})dx_k\cdot g_i(m,a_j,b_j)\right),$$
with $g_i(m,a_j,b_j)$ defined as 
$$\int_{a_i}^{b_i}\varphi_{H_{m,ii}}(x_i-y_{m,i})\tilde{p}(y_j\in[a_j,b_j]|x_i)dx_i,$$
then corrects its prediction using the weighed average of $\{\frac{p_\epsilon(z_r)}{q_\epsilon(z_r)}\}_{r=1}^n$, with $\{z_r\}_{r=1}^n$ i.i.d. samples from $q_\epsilon$ and the weight of each sample point set proportionate to
$$\{\tilde{p}(x_j\in [a_j,b_j]|z_{r,i})\}_{r=1}^n.$$
Because $\tilde{p}$, constructed using histograms, is constant with respect to $x_i$ in each of the intervals $\{[x_{i,s},x_{i,s+1}]\}_{s=0}^{S-1}$, with $S$ the number of partition intervals for the $i$th attribute's range, the above formulas can be evaluated much faster by calculating and caching, for the dimension $i$, each slice $[x_{i,s},x_{i,s+1}]$, and each Gaussian kernel $m$,
$$\int_{x_{i,s}}^{x_{i,s+1}}\varphi_{H_{m,ii}}(z_{m,i}-x_i)dx_i$$
once upon initializing the model, as well as calculating 
$$\{\tilde{p}(y_j\in[a_j,b_j]|x_i\in[x_{i,s},x_{i,s+1}])\}_{s=0}^{S-1}
$$
once upon receiving a query and using the values for every Gaussian kernel.

For our benchmark datasets, the above Bayesnet implementation worked very well on \textbf{power} and \textbf{higgs}, successfully cutting down the max Q-error of queries by more than five times, as the range of the 0th attribute can be narrowed down to an average of 4.5\% of its original by knowing the 4th attribute's value.

\subsection{Introducing ADC+}

We briefly discuss ADC+, another GMM-based addition aimed at enhancing ADC's speed and accuracy on large-volume, high-selectivity queries.

In numerical experiments, we observe that the GMM predictor itself can handle such "easy" queries quite competently, often resulting in a level of precision that rivals or even exceeds the entire prediction-sampling-correction algorithm. Thus, we further upgrade ADC by training a decision tree classifier $\mathrm{T}: \mathbb{R}^2\to \{0,1\}$ and implementing the pseudocode in \textbf{Algorithm 2} (for conciseness, we depict the general case where Bayesnet is not used; in case it is, the algorithm can be easily modified by replacing the above formulas with those in \textbf{5.2}).

ADC+ trains $T$ to minimize the square of log Q-error. That is, given a set of sample queries, we use gini impurity as our loss function and set the weight of each query to $|(\log Q_{ADC})^2-(\log Q_{GMM})^2|$, where $Q_{ADC}$ and $Q_{GMM}$ refers to the Q-error on that query when estimated using the entire ADC algorithm or only the GMM predictor, respectively. To prevent overfitting and reduce storage costs, the decision tree is designed to always have a maximum depth of 4.

As the advantages of ADC+ are more easily verified experimentally instead of theoretically, more information regarding our motives behind creating ADC+, as well as the performance gains resulting from it, shall be provided in subsection \textbf{6.3}.
\begin{algorithm}[t!]
\caption{ADC+}
\label{Algorithm2}
\begin{algorithmic}[1] 
\Require Relation $R$ with attributes $\{A_i\}_{i=1}^d$ and tuples distributed according to $p_0$, GMM predictor $q$, score estimator $s_\theta(x,t)\approx \nabla\log p_t(x)$, density estimator $p_{\theta,\epsilon}$, classifier $\mathrm{T}$, sample number $n$, queried region $V=\prod_{i=1}^d[a_i,b_i]$.
\Ensure $Sel \approx \int_V p_0(x)dx$
\State $\mathrm{vol}\gets \prod_{i=1}^d\frac{b_i-a_i}{\mathrm{max}\{A_i\}-\mathrm{min}\{A_i\}}$
\State $Q\gets \int_Vq(x)dx$
\If {$\mathrm{T}(\log Q, \log \mathrm{vol})=0$} 
\State $Sel \gets Q$
\Else
\State Sample $\{x_r\}_{r=1}^n\sim q_\epsilon(x|x\in V)$
\State $\{y_r\}_{r=1}^n\gets\{\frac{p_{s_\theta,\epsilon}(x_r)}{q_\epsilon(x_r)}\}_{r=1}^n$
\State $Sel \gets Q\cdot \mathrm{mean}\{y_r\}_{r=1}^n$
\EndIf
\State \Return $Sel$
\end{algorithmic}
\end{algorithm}
\section{Experiments}
\subsection{Experimental Setup}
\subsubsection{Datasets}
We test the performance of ADC on three real-world and one synthetic dataset, including:

\textbf{Forest}\cite{reftwentythree}: Forest coverage type dataset with 581012 rows, used as a benchmark by multiple influential works on cardinality estimation \cite{refeighteen} \cite{reffourteen} \cite{refseventeen} \cite{reffive}. Following the practice of \cite{reffive}, we keep the first 10 continuous attributes for evaluation.

\textbf{Power}\cite{reftwentythree}: Dataset with 2.1 million rows, recording power consumption information of a French household at a 1-minute sampling rate, used alongside \textbf{forest} in all four essays above. Following the practice of \cite{reffive}, we keep the 7 continuous measurement attributes excluding date and time, and convert "?" to -1 when dealing with the 1.25\% of rows consisting entirely of missing values. 

\textbf{Higgs}\cite{reftwentythree}: A dataset with 11 million rows and 28 continuous attributes, recording the kinematic properties of microscopic particles measured by detectors in an accelerator, used as a cardinality estimation benchmark by \cite{refthirtyfive} and used by \cite{refthirtysix} to train their summarization model in identifying common frequency and correlation patterns. Following the practice of \cite{refthirtyfive}, we use the 7 high-level features to construct our test dataset \footnote{We round all values in higgs to three-digit decimals because Naru, due to the encoding strategy chosen by \cite{reffive}, can take weeks for training and minutes for inference when given the original data. We note that ADC (which lacks an encoder) and most other models can train and infer just as smoothly on the original dataset.}.

\textbf{Modulo}: A synthetic dataset designed according to the philosophy that leads to our creation of ADC: treating all attributes as a single, unified entity is at times necessary for producing decent estimates. The dataset consists of 4 million rows and five attributes $A$,$B$,$C$,$D$, and $E$, where $A$, $B$, and $D$ are independent integers taking the values 0 to 1999 with equal probability; 
$$C=(A+B+\epsilon_1)\bmod 2000, E=(A+D+\epsilon_2)\bmod 2000,$$ 
and $\epsilon_1$, $\epsilon_2$ are independent integers taking the values 0 to 199 with equal probability. A property of \textbf{modulo} is that all five attributes are \textit{pairwise independent} from one another, making it challenging for estimators that work by discerning pairwise correlations.

We note that \textbf{modulo} is artificially and specifically designed to highlight a correlation pattern that conventional estimators fail to capture, and the results on this dataset should best be interpreted as "We found a rare but solid loophole that no existing model, except for ours, managed to address", rather than "ADC can outperform other models in real-world scenarios as much as it did here".
\subsubsection{Workloads}
For each dataset, we create 100,000 queries for training query-driven methods, 10,000 queries for validation, and 10,000 queries for testing using the program provided by \cite{reffive}. That is, for a $d$-dimensional dataset, we design queries asking for all points lying within a hyperrectangle, where:
\begin{itemize}
\item The hyperrectangle works as a filter for $k$ attributes and spans the entire range of all others, with $k$ a random integer uniformly selected between 1 and $d$.
\item The center of the hyperrectangle has a 90\% chance of being randomly sampled from the dataset (data-entric), and a 10\% chance of being randomly sampled from a uniform distribution spanning the range of each attribute (random-centric)
\item For each attribute with a non-trivial range bound, the query range width on that attribute has a 50\% chance of being sampled from $\mathrm{U}([0,R_i])$ and a 50\% chance of being sampled from $\mathrm{Exp(\frac{10}{R_i})}$, with $R_i$ denoting the difference between $\max$ and $\min$ values of that attribute.
\end{itemize}
\subsubsection{Comparative Techniques}
We briefly introduce versions of our model and the models we use for comparison, which include all learned estimators tested in \cite{reffive}.

For the sake of giving each model a similar size budget while preventing poor results due to bad parameter choices, comparative techniques are trained using the same hyperparameter settings the author of \cite{reffive} chose for the dataset \textbf{forest}, putting the parameter size of each model at around 0.5MB per dataset except Naru on \textbf{higgs} (due to its expensive encoding scheme).
\begin{itemize}
\item \textbf{ADC}: Our estimator ADC in its raw form, always going through the prediction-sampling-correction procedure for every query.

\item \textbf{ADC+}: ADC in its most optimal form, after implementing the changes proposed in \textbf{5.3}.

\item \textbf{ADC-}: The performance of ADC's selectivity estimator as a standalone model, which we report to quantify accuracy improvements gained by the diffusion model corrector, i.e., the score estimator and density estimator.

\item \textbf{MSCN} \cite{refseventeen}: a query-driven regression model which processes the feature vector of queries using a multi-set convolutional network. 

\item \textbf{Lw-Tree} and \textbf{LW-NN} \cite{refeighteen}: two lightweight query-driven regression models that process the feature vector of queries using either a gradient boosted tree (LW-Tree) or a neural network (LW-NN), and enhance their precision with the help of CE (correlation based) features extracted from simple heuristic estimators.

\item \textbf{DeepDB} \cite{refthirtythree}: a data-driven joint distribution model that reconstructs data distributions using sum-product networks.

\item \textbf{Naru} \cite{refsix}: a data-driven joint distribution model that reconstructs data distributions by training an autoregressive model and answers ranged queries using sequential sampling, shown by \cite{reffive} to be the most robust and accurate learned model on most test datasets.
\end{itemize}
\subsubsection{Experiment Device and Accuracy Metric}
We run all of our experiments on a server equipped with dual-socket Intel Xeon Gold 6140 systems, each Intel Xeon Gold 6140 having 36 physical cores (72 hardware threads) and a base clock frequency of 2.30 GHz. We choose Q-error, defined as
$$\max\{\frac{Card_{real}}{Card_{est}},\frac{Card_{est}}{Card_{real}}\}$$
to be our accuracy metric, with zero values for $Card_{real}$ and $Card_{est}$ both set to 1 to prevent division by zero. We report $\max$, 99th percentile, 95th percentile, median, and geometric mean Q-error as our evaluation benchmarks in accordance with \cite{reffive}.

We note that the performance of LW-Tree/LW-NN on \textbf{forest} and MSCN, Naru, DeepDB on \textbf{power} are not as good in our experiments as they were in \cite{reffive}'s, possibly due to different Python library versions and our smaller size budget for the "power" dataset \footnote{In a personal communication, professor Xiaoying Wang, the original author of \cite{reffive}, confirmed that such deviations are plausible and likely attributable to stochastic differences arising from software library versions. They also said that in their original experiments, their GM Q-error results were 1.32 for LW-Tree and 1.35 for LW-NN}. While we record in \textbf{Table 2} the results from our experiments for methodological consistency, we also encourage readers to look up the original figures in \cite{reffive} for a more comprehensive comparison.
\subsection{Experimental Results}
\begin{table*}[t!]
  \caption{Comprehensive Experimental Data.}
  \label{tab:all_data}
  \centering
  \begin{subtable}{\textwidth}
    \centering
    \caption{Q-Error of Different Estimators on Real-world Datasets}
    \label{tab:q_error_all}
    \begin{tabularx}{0.95\textwidth}{lccccccccccccccc}
        \toprule
        Models&
        \multicolumn{5}{c}{Forest} & \multicolumn{5}{c}{Power} & \multicolumn{5}{c}{Higgs} \\     &GM&50th&95th&99th&max&GM&50th&95th&99th&max&GM&50th&95th&99th&max\\
        \midrule
        MSCN&1.72&1.17&8.00&19.5&\textbf{138}&1.26&1.02&5.00&19.9&392&1.20&1.01&2.99&14.8&199 \\
        DeepDB&1.35&\textbf{1.06}&4.25&13.0&3408&1.09&\textbf{1.00}&1.53&3.88&1870&1.06&1.00&1.28&2.77&866\\
        LW-NN&1.51&1.22&4.37&14.33&8306&1.17&1.06&1.89&4.25&$2.2\cdot 10^4$&1.10&1.02&1.50&3.57&$2.6\cdot 10^4$\\
        LW-Tree&1.64&1.33&5.50&13.90&3176&1.13&1.02&1.75&4.33&$2.6\cdot 10^4$&1.09&1.01&1.52&4.17&928\\
        Naru&\textbf{1.27}&\textbf{1.06}&3.21&8.00&452&1.05&1.01&\textbf{1.19}&\textbf{2.00}&3174&1.04&\textbf{1.00}&\textbf{1.11}&2.06&258\\
        ADC- (ours)&1.35&1.11&4.00&8.00&1978&1.05&\textbf{1.00}&1.23&2.24&190&\textbf{1.03}&\textbf{1.00}&1.13&\textbf{1.86}&143\\
        ADC (ours)&1.32&1.16&\textbf{3.00}&\textbf{6.00}&426&1.07&1.03&1.27&\textbf{2.00}&\textbf{128}&1.05&1.01&1.14&2.00&\textbf{78}\\
        ADC+ (ours)&1.28&1.10&\textbf{3.00}&\textbf{6.00}&426&\textbf{1.04}&\textbf{1.00}&1.22&\textbf{2.00}&\textbf{128}&\textbf{1.03}&\textbf{1.00}&1.13&1.96&\textbf{78}\\
        \bottomrule
    \end{tabularx}
  \end{subtable}
  \\
  \vspace{2mm}
  \begin{subtable}{0.33\textwidth}
    \centering
    \caption{Q-error of Different Estimators on Modulo}
    \label{tab:q_error_modulo}
    \begin{tabular}{lccccc}
        \toprule
        Models&\multicolumn{5}{c}{Modulo} \\     &GM&50th&95th&99th&max\\
        \midrule
        
        MSCN&1.74&1.10&16.5&48.0&176\\
        DeepDB&1.55&1.03&9.00&23.0&10503\\
        LW-NN&1.14&1.06&1.54&3.00&919\\
        LW-Tree&1.27&1.07&2.36&17.0&10318\\
        Naru&1.80&\textbf{1.01}&20.57&51.38&30235\\
        ADC-&1.11&\textbf{1.01}&1.74&3.12&70\\
        ADC&1.12&1.07&1.45&\textbf{2.06}&\textbf{40}\\
        ADC+&\textbf{1.08}&\textbf{1.01}&\textbf{1.44}&\textbf{2.06}&\textbf{40}\\
        \bottomrule
    \end{tabular}
  \end{subtable}
  \hfill 
  \begin{subtable}{0.635\textwidth}
    \centering
    \caption{Latency (L,ms), Training Time (T,min), and Storage Cost (S,MB)}
    \label{tab:time_space_info}
    \begin{tabular}{lccccccccccc}
      \toprule
       \multicolumn{3}{c}{Forest} & \multicolumn{3}{c}{Power} & \multicolumn{3}{c}{Higgs} & \multicolumn{3}{c}{Modulo} \\
      L&T&S&L&T&S&L&T&S&L&T&S\\
      \midrule
      0.49&20&0.61&0.36&20&0.54&0.37&20&0.54&0.46&20&0.50\\
      14&12&0.56&4.8&5.7&0.54&4.4&6.9&1.5&4.3&6.2&0.67\\
      0.26&102&0.43&0.33&103&0.43&0.27&102&0.43&0.33&104&0.42\\
      0.12&0.18&0.65\footnote{The open-source code of \cite{reffive} did not report LW-Tree's parameter size like it did with the other four models; thus, the parameter size of LW-Tree is estimated under the assumption that LW-Tree's full storage cost/parameter size ratio is roughly identical to that of LW-NN.}&0.12&0.15&0.60&0.12&0.18&0.57&0.12&0.16&0.59\\
      51&49&0.65&26&116&0.38&155&1770&1.6&34&156&0.43\\
      3.0&600&0.11&4.9&600&0.11&4.4&600&0.11&2.7&600&0.07\\
      34&600&0.44&23&600&0.33&23&600&0.33&16&600&0.22\\
      26&600&0.44&7.9&600&0.33&4.8&600&0.33&7.2&600&0.22\\

      \bottomrule
    \end{tabular}
  \end{subtable}
\end{table*}

\subsubsection{Accuracy Evaluation}
The results for all estimators are listed in \textbf{Table 2(a)} for real-world datasets and \textbf{Table 2(b)} for \textbf{modulo}, which we shall separately discuss.

For the real-world datasets, ADC+ consistently demonstrates top-tier performance, beating every model besides Naru on every benchmark except "max" and "median" on \textbf{forest}. Naru is the only model with performance rivaling ADC+'s, ranking above ADC+ on four benchmarks but falling slightly behind on nine others. Meanwhile, ADC performs almost identically to ADC+ on 95th and 99th Q-error, but falls somewhat behind on GM Q-error due to its higher median Q-error values. Considering that the estimator Naru comes with the downside of much higher latency, such results are enough to earn ADC+ a place among the best cardinality estimators.

For the dataset \textbf{modulo}, ADC's advantage scenario, the complex, multi-attribute correlations caused Naru, DeepDB, and MSCN to significantly drop in accuracy, while LW-Tree's performance exhibited less drastic declines. LW-NN is the only estimator to exhibit performance comparable to ADC on this dataset, almost catching up with ADC+ on 95th and 99th Q-error, but still falling significantly behind ADC+ on the benchmarks GM and max. Meanwhile, the GMM predictor of ADC, a close relative of the KDE algorithm, only exhibits barely satisfactory performance on 95th and 99th Q-error, while Bayesnet implemented via the Chow-Liu algorithm, which only captures pairwise correlations, would be incapable of learning anything more than single-attribute histograms, showing that this decline is not unique to learned estimators.

Compared to its standalone GMM predictor, ADC's diffusion model components result in significant accuracy gains for challenging queries and workloads, but generally fail to further enhance accuracy if the workload and/or query is already simple. That is, ADC significantly reduced max Q-error on all datasets, notably reduced 99th percentile Q-error on all datasets except for \textbf{higgs}, which is the most simple; reduced 95th percentile Q-error for the hardest datasets \textbf{forest} and \textbf{modulo}, but \textbf{increased} median Q-error on all four datasets tested. The reasons why this is the case, as well as what to do about this phenomenon, shall be discussed in subsection \textbf{6.3}.
\subsubsection{Latency, Storage Space, and Training Time}
To evaluate the practicality of ADC+ in actual user scenarios, we discuss the inference latency, training time, and storage requirements of ADC+ in comparison with other learned models.

Regarding per-query latency, the most important feature besides accuracy, ADC+ usually answers a query within the 7-10 milliseconds range on all datasets except \textbf{forest}. This is much slower than MSCN, LW-Tree and LW-NN, which answer a query within 0.5 ms (within 0.2ms for LW-Tree), and somewhat slower than DeepDB, but is also $2\times$ as fast as Naru on \textbf{forest} and \textbf{power} while more than $4\times$ as fast as Naru on \textbf{higgs} and \textbf{modulo}. Given that ADC+ is almost entirely written in Pytorch, it might also be possible to further speed up ADC+ with the help of GPUs, which, sadly, we do not have access to.

In terms of storage space, ADC+ is at an advantage compared to all other models, with parameters taking up 0.44MB of space for \textbf{forest} and less than 0.35MB for all other datasets. To be more precise, the network $s_{\theta_1,head}$ usually takes up 125 to 200 KB of storage, $s_{\theta_2,tail}$ around 50 KB, the GMM predictor usually takes up around $10d$ KB with $d$ the dataset's dimensionality, the Bayesnet component, if used, usually takes up around 40KB, and other artifacts, including the decision tree, some database information, and a histogram we use for single-attribute queries, usually take up around 20 KB. Meanwhile, all comparative techniques have at least 0.42MB of trainable parameters (except Naru on \textbf{power}), and the actual number often goes further beyond that. We note that other learned models can be made more lightweight if needed, though often at a cost. For example, in their original essay, LW-Tree and LW-NN \cite{refeighteen} conducted their experiments using models as small as 16KB, but, as a result, measured significantly higher Q-errors on \textbf{forest}, \textbf{power}, and \textbf{higgs} than we did.

ADC is rather at a disadvantage when it comes to training, with its total training time amounting to 10 hours and only forcibly exiting from training then due to meeting our time constraint. While Naru and LW-NN can also take several hours to train, which is especially true for large datasets, training usually takes as little as 20 minutes for MSCN, 10 minutes for DeepDB, and is completed within one minute for LW-Tree. However, we also point out that the training process of ADC, which heavily relies on Pytorch, can be greatly sped up with the help of a GPU (which we don't have access to), and we are currently working on enhanced score matching methods to cut down training time while producing just as accurate results.
\begin{figure}[h]
  \label{Figure 2}
  \centering
  \includegraphics[width=\linewidth]{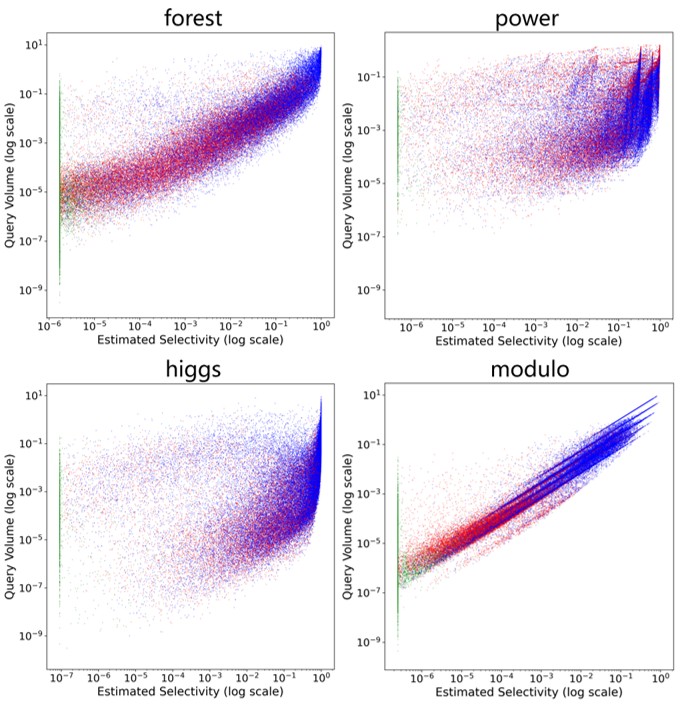}
  \caption{Performance of ADC and ADC- with respect to GMM-estimated selectivity ($x$-axis) and query volume ($y$-axis), both shown on a log scale}
  \Description{Four scatterplots documenting the GMM-estimated selectivity ($x$-axis) and query volume ($y$-axis) of 100,000 training queries each, with values depicted on a log scale, red points being queries where ADC outperforms its GMM predictor, and blue points being queries where the opposite holds. Each of the four graphs has more red points in the lower-left region and a cluster containing almost exclusively blue points in the upper-right region.}
\end{figure}
\subsection{Reducing Median Q-error: Why ADC+}
Our motivation for creating ADC+ is that ADC, in its raw form, performs somewhat awkwardly on median Q-Error, nearly always losing to its GMM predictor, ADC-. In fact, on \textbf{modulo}, ADC ironically increased the distance between median Q-error and 1 by five times, ranking embarrassingly near bottom place on this benchmark despite its decent 95th, 99th, and max Q-error performance.

Why is this? We suspect that it's caused by the fact that, unlike other estimators, there are no "very easy" queries for ADC, and even a query that only asks about two independent attributes or one whose queried region spans 90\% of the sample space must undergo the same prediction-sampling-correction procedure. In this case, ADC's predictor-corrector integration scheme often finds itself adding non-negligible variance to correct an already tiny bias, a choice that is increasingly likely to backfire as the GMM predictor grows more competent.

To test our hypothesis, we create a scatterplot documenting the query volume and GMM-predicted selectivity of 100,000 training queries for each of the 4 datasets, with blue points depicting queries where the GMM predictor performs better than ADC, red points depicting the opposite, and green depicting queries where the two methods are tied (often due to both predicting a cardinality smaller than 1). Indeed, we can see from \textbf{Figure 3} that while the relative proportion of red points tends to increase notably in the lower-left corner, each graph's upper-right corner contains a very dense cluster that is almost exclusively blue.

What's even better, separating a large proportion of blue points from red, at least those on the upper-right corner, can often be done by drawing simple, straight boundary lines, a task ADC performs by training a decision tree. Compared to neural networks, a decision tree of depth 4 takes up less than 3KB and classifies a query in 20 microseconds, making its implementation almost free of costs.

One might worry whether wrongly classifying points that actually benefit from predictor-corrector integration would lead to worse 95th and 99th percentile Q-error results. Fortunately, experiments show that ADC+'s effects on these benchmarks, though small, are seldom negative. On median Q-error, meanwhile, ADC+ acts as the game changer, successfully \textit{halving} the distance between median Q-error and 1 on all \textit{four} datasets while successfully beating MSCN, LW-Tree, and LW-NN on \textbf{power}, \textbf{higgs}, and \textbf{modulo}. What's even better, ADC+ is also able to cut \textit{down} the per-query latency of ADC, because the benefits gained from potentially avoiding two lengthy steps of numerical integration often far outweigh the negligible classification costs.

\section{Conclusion}
The ability of diffusion models to reconstruct complex distributions with high dimensions makes them accurate cardinality estimators. We have proven this by constructing ADC and ADC+, whose accuracy exceeds many famous peers on real-world and synthetic datasets. 

The current capabilities of ADC+ are slightly limited by its long training time and inability to handle categorical attributes, issues we seek to resolve in future studies. We also wonder whether some theorems, network structures, and modeling approaches we proposed upon designing ADC, such as switching from noise to data prediction as $t$ gets larger, might be beneficial for training diffusion models in general, despite influences of the manifold hypothesis and the much higher (3000+) dimensionality. 

Overall, ADC+ provides a good starting point for exploring the use of diffusion models in cardinality estimation, a rather unconventional crossover that can produce very promising fruits.
\appendix
\section{Proof of Theorems}
\begin{theorem 2.1}
\label{thm0}
Let $p_0$ be a probabilistic distribution and $\{p_t|t\in[0,T]\}$ be a family of distributions derived by perturbing $p_0$ using the SDE (\ref{forward}) from time 0 to time $T$. Let $k(t)=e^{\int_0^t\alpha(s)ds}$, and $\sigma^2(t)$ be the solution of the initial value Ordinary Differential Equation (ODE)
\begin{align*}
{d\sigma^2}/{dt}&=\beta(t)-2\alpha(t)\sigma^2(t)\\
\sigma^2(0)&=0,
\end{align*}
then, for each choice of $\alpha$ and $\beta$, it would hold that
\begin{align}
p_{t,\alpha,\beta}&=(\tau_{k(t)}\circ p_0)*\varphi_{\sigma^2(t)}\\
\nabla \log p_{t,\alpha,\beta}(x)&=k(t)\nabla \log p_{\sigma^2(t)k^2(t),0,1}(k(t)x),
\end{align}
where $\tau$ in $\mathbb{R}^d$ is the scaling operator $(\tau_k\circ f)(x)=k^df(kx)$ and $\varphi_{\sigma^2(t)}$ in $\mathbb{R}^d$ is the probability density function of $\mathcal{N}(0,\sigma^2(t)I_d)$.
\end{theorem 2.1}
\begin{proof}
It can be easily verified that under perturbation of the stochastic differential equation 
$dx=-\alpha(t)xdt+\beta(t)dw,$
a random variable $x(0)$ which takes the value $x_0$ with probability 1 would satisfy the distribution $x(t)\sim\mathcal{N}(\frac{x_0}{k(t)},\sigma^2(t))$ at time $t$. Therefore, if $x_0\sim p_0(x)$,$p_t(x)$ would have the expression 
\begin{align*}
p_{t,\alpha,\beta}(x)&=\int_\mathbb{R^d}(2\pi\sigma^2(t))^{-\frac{d}{2}}e^{-\frac{\left\|x-\frac{y}{k(t)}\right\|_2^2}{2\sigma^2(t)}}p_0(y)dy
\\&=\int_\mathbb{R^d}(2\pi\sigma^2(t))^{-\frac{d}{2}}e^{-\frac{\left\|x-\frac{y}{k(t)}\right\|_2^2}{2\sigma^2(t)}}k^n(t)p_0(y)d(\frac{y}{k(t)}),
\end{align*}
hence $p_{t,\alpha,\beta}=(\tau_{k(t)}p_0)*\mathcal{N}(0,\sigma^2(t))$.\\
As 
$$(\tau_\lambda\circ f)*(\tau_\lambda \circ g)=\tau_\lambda\circ (f*g)$$ and $$\tau_\lambda\circ \mathcal{N}(0,\sigma^2)=\mathcal{N}(0,\frac{\sigma^2}{\lambda^2}),$$
this means that 
$$p_{t,\alpha,\beta}(x)=\tau_{k(t)}\circ\big(p_0*\mathcal{N}(0,k^2(t)\sigma^2(t))\big),$$
and hence $p_{t,\alpha,\beta}(x)=\big(\tau_k\circ p_{\sigma^2k^2,0,1}\big)(x)$.\\
Directly differentiating the formula on both sides immediately gives
$$\nabla\log p_{t,\alpha,\beta}(x)=k(t)\nabla\log p_{\sigma^2k^2,0,1}(k(t)x),$$
finishing the proof.
\end{proof}
\begin{theorem 2.2}
Let $s_\theta(x,t):\mathbb{R}^d\times\mathbb{R^+}\to \mathbb{R}^d$ be a neural network and $p_0$ be a probability distribution, and $p_{0,\theta}(x)$ to be the marginal distribution of $\tilde{x}(0)$ after $\tilde{x}(T)\sim\mathcal{N}(0,\sigma^2(T)I_d)$ is perturbed by the reverse diffusion SDE
$$dx=-[\alpha(t)x+\beta^2(t)s_\theta(x,t)]dt+\beta(t)d\hat{w}$$
from time $t$ to time $0$. Using the notation of Theorem \ref{thm0}, the KL divergence between $p_0$ and $p_\theta$ can be upper bounded by
\begin{align}
\label{KLBound}
D_{KL}\left(p_0\left(x\right),p_{0,\theta}\left(x\right)\right)\leq D_{KL}\left(p_{T,\alpha,\beta},\mathcal{N}\left(0,\sigma^2(T)I_d\right)\right) \notag \\
+\frac{1}{2}\int_0^T\mathrm{E}_{p_{t,\alpha,\beta}(x)}[\beta^2(t)\left\|\nabla \log p_{t,\alpha,\beta}(x)-s_\theta (x,t)\right\|_2^2]dt,
\end{align}
while the value of $\log p_\theta (x_0)$ can be lower bounded by
$$
\log p_{0,\theta}(x_0)\geq\mathrm{E}_{p_{0T}(x|x_0)}\log \varphi_{\sigma^2(T)}(x)+\int_0^T n\alpha(t)dt-\frac{1}{2}\int_0^T\beta^2(t)\cdot 
$$
\begin{align}
\label{density}
\mathrm{E}_{p_{0t(x|x_0)}}&[\left\|\nabla \log p_{0t}(x|x_0)-s_\theta(x,t)\right\|_2^2-\left\|\nabla \log p_{0t}(x|x_0)\right\|_2^2]dt
\end{align}
with the equal sign holding in both inequalities if there exists a probability distribution $q$ such that $s_\theta(x,t)=\nabla\log q_t(x), \forall t \in [0,T]$. 
\end{theorem 2.2}
\begin{proof}
Proof for this theorem is elegantly provided by Song, Durkan, Murray and Ermon in \cite{reftwo}.
\end{proof}
\begin{theorem 3.1}
Let $p_0$ be the weighted average of several positive distributions $\{p_i\}_{i=1}^n$, such that
$$\{V_i\}_{i=1}^n:=\{x|p_i(x)>0\}_{i=1}^n$$
are open domains in $\mathbb{R}^d$ intersecting only at their boundaries, and each $p_i$ is Lipschitz with constant $k_i$ in the interior (but not necessarily at the boundary) of $V_i$. Denote $V_0=\mathbb{R}^d\textbackslash\cup_{i=1}^n V_i$
for convenience, and let $\{p_t|t\in(0,T]\}$ be the family of distributions derived by perturbing $p_0$ according to the VE\cite{refone} scheme (results for the VP\cite{refone} scheme can be immediately derived from Theorem \ref{thm0}), then:
\begin{enumerate}
\item $\forall x\in V_i$ such that $B(x,r)\subseteq V_i$, 
$$\nabla \log p_t(x_0)=\nabla \log p_0(x_0)+\delta_1(x_0,t),$$
with $\delta_1(x_0,t)$ an infinitesimal term of order $O(\sigma(t))$.
\item $\forall x \in \partial V_i\cap\partial V_j$ with at most one of $i$ and $j$ being zero, if $B(x_0,r)\subseteq V_i\cup V_j$, $p_{0,i}(x_0)+p_{0,j}(x_0)>0$, and inside $B(x_0,r)$, 
$$\mathrm{d}\left(y,T_{x_0}(\partial V_i\cap \partial V_j)\right)\leq h\left\|y-x_0\right\|_2^2$$ holds for some constant $h$ and all $y\in\partial V_i\cap\partial V_j$, then $\forall \lambda$, $\nabla \log p_t(x_0+\lambda\sqrt{t}\vec{n})$ would equal
$$
\frac{e^{-\frac{\lambda^2}{2}}(p_{0,i}(x_0)-p_{0,j}(x_0))\vec{n}}{2\sigma(t)\left(\Phi(\lambda)p_{0,i}(x_0)+(1-\Phi(\lambda)p_{0,j}(x_0))\right)}+\delta_\sigma(x_0,\lambda,t),
$$
with $\vec{n}$ the unit vector normal to $T_{x_0}(\partial V_i\cap\partial V_j)$, pointing from $V_j$ to $V_i$, and $\delta_\sigma(x_0,\lambda,t)$ a residual term of order $O(1)$.
\item $\forall x\in V_0$, suppose $y_0\in \bar{V_i}$ is the point closest to $x_0$ such that $y_0\in\cup_{i=1}^k\bar{V_i}$. If $P(x\in B_r(y)|x\sim p_0)>M_1r^k, \forall 0<r<R$ for some constants $M_1, k_1$, while
$$\left\|z-x_0\right\|_2^2-\left\|y_0-x_0\right\|_2^2\geq \min\{M_2,k_2\left\|z-y_0\right\|_2^2\}$$
for all $z\in\cup_{i=1}^kV_i$ and some constants $M_2,k_2$, then
$$\nabla\log p_t(x_0)=\frac{y_0-x_0}{\sigma^2(t)}+\delta_{\sigma^2}(x_0,t),$$
with $\delta_{\sigma^2}(x_0,t)$ a term of order $o(\frac{1}{\sigma(t)^{1+\epsilon}})$ for all $\epsilon>0$.
\end{enumerate}
\end{theorem 3.1}
\begin{proof}[Proof of 3.1(1)]
With a change of coordinates, we may assume WLOG that $x_0=0$ and that $\nabla p_0(x)=ke_1$ is parallel to the first coordinate axis. Now, consider the cube $C_r=\{x|\left\|x-x_0\right\|_1<r\}$ where $r=\frac{R}{d}$, and we would have
$$\nabla\log p_t(x)=\frac{\int_{\mathbb{R}^d}p_0(x)(2\pi t)^{-\frac{d}{2}}e^{-\frac{\left\|x\right\|_2^2}{2t}}\cdot\frac{x}{t}dx}{\int_{\mathbb{R}^d}p_0(x)(2\pi t)^{-\frac{d}{2}}e^{-\frac{\left\|x\right\|_2^2}{2t}}dx}$$
which can be decomposed into
$$\frac{\int_{C_r}p_0(x)(2\pi t)^{-\frac{d}{2}}e^{-\frac{\left\|x\right\|_2^2}{2t}}\frac{x}{t}dx+\int_{\mathbb{R}^d\backslash C_r}p_0(x)(2\pi t)^{-\frac{d}{2}}e^{-\frac{\left\|x\right\|_2^2}{2t}}\frac{x}{t}dx}{\int_{C_r}p_0(x)(2\pi t)^{-\frac{d}{2}}e^{-\frac{\left\|x\right\|_2^2}{2t}}dx+\int_{\mathbb{R}^d\backslash C_r}p_0(x)(2\pi t)^{-\frac{d}{2}}e^{-\frac{\left\|x\right\|_2^2}{2t}}dx}$$
We shall now denote the term $\int_{C_r}p_0(x)(2\pi t)^{-\frac{d}{2}}e^{-\frac{\left\|x\right\|_2^2}{2t}}\frac{x}{t}dx$ by $a_1$, $\int_{\mathbb{R}^d\backslash C_r}p_0(x)(2\pi t)^{-\frac{d}{2}}e^{-\frac{\left\|x\right\|_2^2}{2t}}\frac{x}{t}dx$ by $a_2$; $\int_{C_r}p_0(x)(2\pi t)^{-\frac{d}{2}}e^{-\frac{\left\|x\right\|_2^2}{2t}}dx$ by $b_1$, $\int_{\mathbb{R}^d\backslash C_r}p_0(x)(2\pi t)^{-\frac{d}{2}}e^{-\frac{\left\|x\right\|_2^2}{2t}}dx$ by $b_2$. Now, it suffices to show
\begin{itemize}
\item $a_1=\nabla p_0(x_0)+O(\sqrt{t})$ 
\item $a_2=o(t^n)$ for all $n$ 
\item $b_1=p_0(x_0)+O(\sqrt{t})$
\item $b_2=o(t^n)$ for all $n$.
\end{itemize}
The bounds on $a_2$ and $b_2$ are easy to verify, as $\int_{\mathbb{R}^d\backslash C_r}p_0(x)dx\leq 1$ and $e^{-\frac{\left\|x\right\|_2^2}{2t}}\cdot\frac{x}{t^{n+1}}$ decays uniformly to zero in $\mathbb{R}^d\backslash C_r$ for all $n$ as $t\to 0$. For the estimation of $a_1$, we have 
\begin{align*}
a_1&=\int_{C_r}p_0(x)(2\pi t)^{-\frac{d}{2}}e^{-\frac{\left\|x\right\|_2^2}{2t}} \frac{x}{t}dx\\
&=\frac{1}{2}\int_{C_r}\big(p_0(x)-p_0(-x)\big)(2\pi t)^{-\frac{d}{2}}e^{-\frac{\left\|x\right\|_2^2}{2t}}\frac{x}{t}dx\\
&=\int_{C_r}\Big(\langle\nabla p_0(x_0),x\rangle+\frac{1}{2}\big(\epsilon_{p_0}(x)-\epsilon_{p_0}(-x)\big)\Big)(2\pi t)^{-\frac{d}{2}}e^{-\frac{\left\|x\right\|_2^2}{2t}} \frac{x}{t}dx,
\end{align*}
where $\epsilon_{p_0}(x)=p_0(x)-p_0(x_0)-\langle\nabla p_0(x_0),x-x_0\rangle$.\\
Using the fact that $\nabla p_0(x_0)=ke_1$, we know that for $n>1$,
$$\int_{C_r}\langle\nabla p_0(x_0),x\rangle(2\pi t)^{-\frac{d}{2}}e^{-\frac{\left\|x\right\|_2^2}{2t}} \frac{\langle x,v_n\rangle}{t}dx=0,$$
and for $n=1$,
\begin{align*}
&\int_{C_r}\langle\nabla p_0(x_0),x\rangle(2\pi t)^{-\frac{d}{2}}e^{-\frac{\left\|x\right\|_2^2}{2t}} \frac{\langle x,v_1\rangle}{t}dx\\
=&\int_{-r}^r(2\pi t)^{-\frac{1}{2}}e^{-\frac{x_1^2}{2t}}\frac{kx_1^2}{t}dx_1\int_{C_r^{d-1}}(2\pi t)^{-\frac{d-1}{2}}e^{-\frac{x_2^2+...+x_d^2}{2t}}dx_2...dx_d\\
=&k\int_{-\sqrt{t}^{-1}r}^{\sqrt{t}^{-1}r}(2\pi)^{-\frac{1}{2}}e^{-\frac{x_1^2}{2}}x_1^2dx\cdot\left(2\Phi\left(\sqrt{t}^{-1}r\right)-1\right)^{d-1},
\end{align*}
a term which approaches $k=\nabla p_0(x_0)$ exponentially fast as $t\to 0$.\\
Meanwhile, for the residual term, we have
\begin{align*}
&\left\|\frac{1}{2}\int_{C_r}\big(\epsilon_{p_0}(x)-\epsilon_{p_0}(-x)\big)(2\pi t)^{-\frac{d}{2}}e^{-\frac{\left\|x\right\|_2^2}{2t}}\frac{x}{t}dx\right\|_2\\
\leq&\int_{C_r}\epsilon_{p_0}(x)(2\pi t)^{-\frac{d}{2}}e^{-\frac{\left\|x\right\|_2^2}{2t}}\frac{\left\|x\right\|_2}{t}dx\\
\leq&\int_{\mathbb{R}^d}\frac{h\left\|x\right\|_2^3}{t}(2\pi t)^{-\frac{d}{2}}e^{-\frac{\left\|x\right\|_2^2}{2t}}dx\\
\leq&\sqrt{t}\int_{\mathbb{R}^d}h\left\|x\right\|_2^3(2\pi)^{-\frac{d}{2}}e^{-\frac{\left\|x\right\|_2^2}{2}}dx=O(\sqrt{t}),
\end{align*}
hence $a_1=\nabla p_0(x_0)+O(\sqrt{t}).$ Finally, for the bound on $b_1$, we can similarly verify using the symmetry of $C_r$ that 
$$\int_{C_r}(2\pi t)^{-\frac{d}{2}}e^{-\frac{\left\|x\right\|_2^2}{2t}}p_0(x_0)+\langle x-x_0,\nabla p_0(x_0)\rangle dx$$
equals $p_0(x_0)\cdot\big(2\Phi(r\sqrt{t}-1)\big)^d$, hence $p_0(x_0)$ minus an exponentially decaying term, and that 
\begin{align*}
&\left|\int_{C_r}(2\pi t)^{-\frac{d}{2}}e^{-\frac{\left\|x\right\|_2^2}{2t}}\epsilon_{p_0}(x)dx\right|\\
\leq&\int_{C_r}h\left\|x\right\|_2^2(2\pi t)^{-\frac{d}{2}}e^{-\frac{\left\|x\right\|_2^2}{2t}}dx\leq ht
\end{align*}
hence $b_1=p_0(x_0)+O(t)$, finishing the proof.
\end{proof}
\begin{proof}[Proof of 3.1(2)]
We may assume WLOG that $x_0=0$ and $\vec{n}=e_1$.\\
Defining $C_r$ and $a_1, a_2, b_1, b_2$ as in the proof of 3.1(1), we shall continue to verify that 
\begin{itemize}
\item $a_1=\frac{e^{-\frac{\lambda^2}{2}}\big(p_i(0)-p_j(0)\big)}{\sqrt{2\pi t}}+O(1)$ 
\item $a_2=o(t^n)$ for all $n$
\item $b_1=\Phi(\lambda)p_{0,i}(x_0)+(1-\Phi(\lambda)p_{0,j}(x_0)+O(\sqrt{t})$
\item $b_2=o(t^n)$ for all $n$.
\end{itemize}
As before, the bounds on $a_2$ and $b_2$ can be verified using the uniform convergence of $e^{-\frac{\left\|x\right\|_2^2}{2t}}\cdot\frac{x}{t^{n+1}}$ in $\mathbb{R}^n\backslash C_r$.To derive the estimations on $a_1$ and $b_1$, we further divide $C_r$ into the subsets $V_i^+,V_i^-,V_j^+,V_j^-$, where $V_i^-=V_i\cap C_r\cap \{x|x_1<0\}$, $V_i^+=V_i\cap C_r\cap \{x|x_1\geq 0\}$, and $V_j^-,V_j^+$ are similarly defined.\\
This way, $p_0\backslash C_r$ can be decomposed as $p_0=p_{C_r}+\Delta p+\epsilon_p$, where:
\begin{equation}
p_{C_r}(x) = \begin{cases} 
p_{0,i}(x_0) &  x_1\geq0 \\
p_{0,j}(x_0) &  x_1<0 
\end{cases}
\end{equation}
\begin{equation}
\Delta p(x) = \begin{cases} 
p_{0,i}(x_0)-p_{0,j}(x_0) &  x\in V_i^- \\
p_{0,j}(x_0)-p_{0,i}(x_0) &  x\in V_j^+ \\
0 & \mathrm{else}
\end{cases}
\end{equation}
\begin{equation}
\epsilon_p(x) = \begin{cases} 
p_{0,i}(x)-p_{0,i}(x_0) &  x\in V_i \\
p_{0,j}(x)-p_{0,j}(x_0) &  x\in V_j 
\end{cases}
\end{equation}
Now, for the estimate of $p_{C_r}$, we can verify with the methods used to prove 3.1(1) that 
$\int_{C_r}p_{C_r}(x)(2\pi t)^{-\frac{d}{2}}e^{-\frac{\left\|x-\lambda\sqrt{t}e_1\right\|_2^2}{2t}}dx$ equals
$$\big(\Phi(\lambda)p_{0,i}(x_0)+\big(1-\Phi(\lambda)\big)p_{0,j}(x_0)\big)+o(t),$$
and that $\int_{C_r}p_{C_r}(x)(2\pi t)^{-\frac{d}{2}}e^{-\frac{\left\|x-\lambda\sqrt{t}e_1\right\|_2^2}{2t}}\frac{x}{t}dx$ equals
$$\frac{e^{-\frac{\lambda^2}{2}}(p_i(0)-p_j(0))}{\sqrt{2\pi t}}+o(t).$$
For $\epsilon_p$, we can also verify with the methods used to prove 3.1(1) that 
\begin{align*}
&\left|\int_{C_r}\epsilon_p(x)(2\pi t)^{-\frac{d}{2}}e^{-\frac{\left\|x-\lambda\sqrt{t}e_1\right\|_2^2}{2t}}dx\right|\\
\leq&\int_{\mathbb{R}^d}k\left\|x+\lambda\sqrt{t}e_1\right\|_2(2\pi t)^{-\frac{d}{2}}e^{-\frac{\left\|x\right\|_2^2}{2t}}dx,
\end{align*}
hence belonging to $O(\sqrt{t})$, and that 
\begin{align*}
&\left\|\int_{C_r}\epsilon_p(x)(2\pi t)^{-\frac{d}{2}}e^{-\frac{\left\|x-\lambda\sqrt{t}e_1\right\|_2^2}{2t}}\frac{x-\lambda\sqrt{t}e_1}{t}dx\right\|_2\\
\leq&k\int_{\mathbb{R}^n}(2\pi t)^{-\frac{d}{2}}e^{-\frac{\left\|x-\lambda\sqrt{t}e_1\right\|_2^2}{2t}}\frac{\left\|x\right\|_2\cdot\left(\left\|x\right\|_2+\left\|\lambda\sqrt{t}e_1\right\|_2\right)}{t}dx,
\end{align*}
hence belonging to $O(1)$.\\
Finally, for the estimation of $\Delta p$, defining $C_r^{d-1}=\{(x_2,...,x_n)\}$, we have
\begin{align*}
&\left|\int_{C_r}\Delta p(x)(2\pi t)^{-\frac{d}{2}}e^{-\frac{\left\|x-\lambda\sqrt{t}e_1\right\|_2^2}{2t}}dx\right|_2\\
\leq&\int_{C_r^{d-1}}\int_{-h\left\|y\right\|_2^2}^{h\left\|y\right\|_2^2}\left|\Delta p(x)\right|(2\pi t)^{-\frac{1}{2}}e^{-\frac{(x_1-\lambda\sqrt{t})^2}{2t}}dx_1(2\pi t)^{-\frac{d-1}{2}}e^{-\frac{\left\|y\right\|_2^2}{2t}}dy\\
\leq&\sup_x\left|\Delta p(x)\right|\int_{C_r^{d-1}}\frac{2h\left\|y\right\|_2^2}{\sqrt{2\pi t}}(2\pi t)^{-\frac{d-1}{2}}e^{-\frac{\left\|y\right\|_2^2}{2t}}dy=O(\sqrt{t})
\end{align*}
and that, similarly,
\begin{align*}
&\int_{C_r^{d-1}}\int_{-h\left\|y\right\|_2^2}^{h\left\|y\right\|_2^2}|\Delta p|(2\pi t)^{-\frac{d}{2}}e^{-\frac{\left\|x-\lambda\sqrt{t}e_1\right\|_2^2}{2t}}\frac{\left\|(x_1,y)-\lambda\sqrt{t}e_1\right\|_2}{t}dx_1dy\\
\leq&\sup_x\left|\Delta p(x)\right|\int_{C_r^{d-1}}(2\pi t)^{-\frac{d-1}{2}}e^{-\frac{\left\|y\right\|_2^2}{2t}}\frac{2h\left\|y\right\|_2^2\left(h\left\|y\right\|_2+\lambda\sqrt{t}+\left\|y\right\|_2\right)}{t\sqrt{2\pi t}}dy\\
=&O(1)
\end{align*}
Summing all of the above, we conclude that 
$$a_1=\frac{e^{-\frac{\lambda^2}{2}}\big(p_i(0)-p_j(0)\big)}{\sqrt{2\pi t}}+O(1)$$
and that 
$$b_1=\Phi(\lambda)p_{0,i}(x_0)+(1-\Phi(\lambda)p_{0,j}(x_0)+O(\sqrt{t}),$$
finishing the proof.
\end{proof}
\begin{proof}[Proof of 3.1(3)]
We shall still assume WLOG that $y_0=0$. Define the ball $B_{t,\delta}$ as $\{x| \left\|x-y_0\right\|_2<R\sqrt{t}^{1-\delta}\}$, then, $\nabla \log p_t(x_0)+\frac{x_0}{t}$ would equal 
$$\frac{\int_{B_{t,\delta}}\frac{p_0(x)}{(2\pi t)^{d/2}}e^{-\frac{\left\|x-x_0\right\|_2^2}{2t}}\frac{x}{t}dx+\int_{\mathbb{R}^d\backslash B_{t,\delta}}\frac{p_0(x)}{(2\pi t)^{d/2}}e^{-\frac{\left\|x-x_0\right\|_2^2}{2t}}\frac{x}{t}dx}{\int_{B_{t,\delta}}\frac{p_0(x)}{(2\pi t)^{d/2}}e^{-\frac{\left\|x-x_0\right\|_2^2}{2t}}dx+\int_{\mathbb{R}^d\backslash B_{t,\delta}}\frac{p_0(x)}{(2\pi t)^{d/2}}e^{-\frac{\left\|x-x_0\right\|_2^2}{2t}}dx}.$$
Still denoting the two terms in the numerator as $a_1, a_2$ and the two terms in the denominator as $b_1, b_2$, we now seek to prove that  
\begin{itemize}
\item $a_1<b_1R\sqrt{t}^{-(1+\delta)}$
\item $a_2=o(t^n)e^{-\frac{\left\|x_0\right\|_2^2}{2t}}$ for all $n>0$
\item $b_1>Ct^ke^{-\frac{\left\|x_0\right\|_2^2}{2t}}$
\item $b_2=o(t^n)e^{-\frac{\left\|x_0\right\|_2^2}{2t}}$ for all $n>0$
\end{itemize}
As $e^{-\frac{\left\|x-x_0\right\|_2^2}{2t}}\leq e^{-\frac{\left\|x_0\right\|_2^2}{2t}}\max \{e^{-\frac{k_2t^{1-\delta}R^2}{2t}},e^{-\frac{M_2}{2t}}\}$ on $\mathbb{R^d}\backslash B_{t,\delta}$, the bound for $b_2$ is easy to derive. The bound for $a_2$ can be similarly proven, as for small enough $t$, $\left\|x-x_0\right\|_2>\left\|x\right\|_2-\left\|x_0\right\|_2$ ensures that
$$e^{-\frac{\left\|x-x_0\right\|_2^2}{2t}}\left\|x\right\|_2\leq \max \{e^{-\frac{\left\|x-x_0\right\|_2^2}{2t}},e^{-\frac{\left\|x_0\right\|_2^2}{2t}}\}\cdot2\left\|x_0\right\|_2.$$
To prove the lower bound on $b_1$, consider the ball $B_{t^2,0}$. Noting that 
$$e^{-\frac{\left\|x-x_0\right\|_2^2}{2t}}>e^{-\frac{\left\|x_0\right\|_2^2}{2t}}e^{-\left\|x_0\right\|_2}e^{-\frac{t}{2}}$$ holds for all $x\in B_{t^2,0}$, $P(x\in B_r(y)|x\sim p_0)>M_1r^{k_1}$ implies that 
$$\int_{B_{t^2,0}}(2\pi t)^{-\frac{d}{2}}e^{-\frac{\left\|x-x_0\right\|_2^2}{2t}}p_0(x)dx\geq M_1t^{k_1}\cdot e^{-\frac{\left\|x_0\right\|_2^2}{2t}}e^{-\left\|x_0\right\|_2}e^{-\frac{t}{2}}.$$
Finally, $a_1<b_1R\sqrt{t}^{-(1+\delta)}$ follows immediately from the fact that $\left\|x\right\|_2t^{-1}\leq R{\sqrt{t}^{-(1+\delta)}}$ in $B_{t,\delta}$.
\end{proof}
\begin{theorem 4.1}
Let $\tilde{x}$ be a random variable taking values $\{x_i\}_{1\leq i\leq n}$ with equal probability (as is \textbf{always} the case when we construct $p_t$ from i.i.d. sample points), such that $\frac{1}{n}\sum_{i=1}^nx_i=0$. Take $p_t(x)$ to be the marginal distribution of $x(t)$ when $x(0)=\tilde{x}$ is perturbed by the diffusion SDE
$$dx=-\alpha(t)xdt+\beta(t)dw$$
from time 0 to time $t$. Then, we would have 
\begin{enumerate}
\item
$\mathrm{E}_{p_t(x)}\left\|\nabla\log p_t(x)+\frac{x}{\sigma^2(t)}\right\|_2^2\leq \frac{f(t,\tilde{x})}{\sigma^4(t)k^2(t)}.$
\item
$\mathrm{E}_{p_t(x)}\left\|\nabla \log p_t(x)+\frac{x}{\sigma^2(t)}\right\|_2^2\leq \frac{f(t,\tilde{x})g(t,\tilde{x})}{\sigma^4(t)k^2(t)}.$
\end{enumerate}
where
$$f(t,\tilde{x})=\frac{1}{n}\sum_{i=1}^n\left\|x_i\right\|_2^2$$
and
$$g(t,\tilde{x})=\frac{1}{n}\sum_{i=1}^n(e^{\frac{2\left\|x_i\right\|_2^2+f(t,\tilde{x})}{2\sigma^2(t)k^2(t)}}-1).$$
\end{theorem 4.1}
\begin{proof}[Proof of 4.1(1)]
We only need to prove the lemma under the VE perturbation scheme ($\alpha(t)=0,\beta(t)=1$), and all other cases follow immediately from the equivalence of perturbation schemes.\\
In this case, $p_t(x)$ can be expressed as 
$$p_t(x)=\frac{1}{n}\sum_{i=1}^n\mathcal{N}\big(x_i,t\big),$$
meaning that $\mathrm{E}_{p_t(x)}\left\|\nabla\log p_t(x)+\frac{x}{\sigma^2(t)}\right\|_2^2$ would equal
$$\int_{\mathbb{R}^d}\frac{\left\|\frac{1}{n}\sum_{i=1}^n\mathcal{N}\big(x_i,t\big)(x)\cdot\frac{x_i}{t}\right\|_2^2}{\frac{1}{n}\sum_{i=1}^n\mathcal{N}\big(x_i,t\big)(x)}dx,$$
which, by Jensen's Inequality, is no greater than 
$$\int_{\mathbb{R}^d}\frac{1}{n}\sum_{i=1}^n\Big(\mathcal{N}\big(x_i,t)\big)(x)\left\|\frac{x_i}{t}\right\|_2^2\Big)dx=f(t,\tilde{x}),$$
noting that $k(t)=1$ under the VE scheme.
\end{proof}
\begin{proof}[Proof of 4.1(2)]
As $\frac{1}{n}\sum_{i=1}^n x_i=0$, $$\nabla \log p_t(x)+\frac{x}{\sigma^2(t)}=\frac{\frac{1}{n}\sum_{i=1}^n\big(\mathcal{N}\big(x_j,\sigma^2(t)\big)(x)-k\big)\frac{x_i}{\sigma^2(t)}}{\frac{1}{n}\sum_{i=1}^n\mathcal{N}\big(x_j,\sigma^2(t)\big)(x)}$$
would hold for any $k\in \mathbb{R}$. Thus, $\left\|\nabla\log p_t(x)+\frac{x}{\sigma^2(t)}\right\|_2^2$ would be no greater than
$$\left(\frac{1}{n}\sum_{i=1}^n\left(\frac{\mathcal{N}\big(x_i,\sigma^2(t)\big)(x)-k}{\frac{1}{n}\sum_{j=1}^n\mathcal{N}\big(x_j,\sigma^2(t)\big)(x)}\right)^2\right)\left(\frac{1}{n}\sum_{i=1}^n\left\|\frac{x_i}{\sigma^2(t)}\right\|_2^2\right).$$
In the above formula, the second term is easily calculated as $f(t,\tilde{x})$. To upper bound the first term, we take $k$ to be $\frac{1}{n}\sum_{j=1}^n\mathcal{N}\big(x_j,\sigma^2(t)\big)(x)$, after which the original expression can be proven equal to
$$\int_{\mathbb{R}^d}\frac{1}{n}\sum_{i=1}^n\frac{\big(\mathcal{N}\big(x_i,\sigma^2(t)\big)(x)\big)^2}{\frac{1}{n}\sum_{i=1}^n\mathcal{N}\big(x_j,\sigma^2(t)\big)(x)}dx-1.$$
Using the mean value inequality, $\frac{1}{n}\sum_{i=1}^n\mathcal{N}\big(x_j,\sigma^2(t)\big)(x)$ can be lower bounded by 
$$\big(2\pi\sigma^2(t)\big)^{-\frac{d}{2}}e^{-\frac{\left\|x\right\|_2^2-\frac{2}{n}\sum_{j=1}^n\langle x_j,x\rangle+\frac{1}{n}\sum_{j=1}^n\left\|x_j\right\|_2^2}{2\sigma^2(t)}}.$$
Plugging the above expression into the denominator makes the original expression integrable, resulting in the final estimation $f(t,\tilde{x}),g(t,\tilde{x})$.
\end{proof}

\textbf{Note}: Our model, ADC, estimates the moment of transition from $O(\frac{1}{\sigma^4(t)k^2(t)})$ decay to $O(\frac{1}{\sigma^6(t)k^4(t)})$ decay using $\max \left\|x\right\|_2^2$ rather than $\mathrm{E}_{p_0(x)}\left\|x\right\|_2^2$. This is due to the fact that, for $p_0$ with non-bounded support, such a transition might be arbitrarily postponed even if $\mathrm{E}_{p_0(x)}\left\|x\right\|_2^2$ remains bounded.\\
For example, let $\tilde{x_N}$ take the value $-\frac{1}{N}$ with probability $\frac{N^2}{N^2+1}$, the value $N$ with probability $\frac{1}{N^2+1}$, and $p_{t,N}(x)$ denote the marginal distribution of $x(t)$ when $x(0)=\tilde{x_N}$ is perturbed using the diffusion SDE above. Then, we would have
$$\lim_{N\to\infty}\sigma^4(t)k^2(t)\mathrm{E}_{p_{t,N}(x)}\left\|\nabla\log p_{t,N}(x)+\frac{x}{\sigma^2(t)}\right\|_2^2=1.$$
for all $t>0$.

\begin{theorem 4.2}
Define $\mathrm{Var}(x_0)$ to be the error term's variance upon approximating $\log p_{\epsilon,\theta}(x_0)$ with formula (\ref{QMC}). Assuming that 
\begin{enumerate}
\item The partition is dense enough that for the same $y$, fluctuations of $\Delta t_i\delta_{t,space}(x_0,y)$ for different $t$ in $[t_{i-1},t_{i}]$ is negligible compared to the total integration error,
\item The error terms $\{\delta_{i,time}(x_0)\}_{i=1}^N$ are negligible compared to $\{\delta_{\tilde{t}_i,space}(x_0,y(\tilde{t_i}))\}_{i=1}^N$, and
\item The sequence $\{y(\tilde{t}_i)\}_{i=1}^N$ consists of i.i.d. samples from $\mathrm{U}([0,1]^d)$
\end{enumerate}
hold for all $x_0$, $\mathrm{E}_{x_0\sim p_0}\mathrm{Var}(x_0)$ would be minimized under a timestep scheme where $t_i-t_{i-1}$ is chosen inversely proportionate to $$\big(\mathrm{E}_{x_0\sim p_0}[\mathrm{Var}_{y}\delta_{\tilde{t}_i,space}(x_0,y)]\big)^{\frac{1}{2}}.$$
\end{theorem 4.2}
\begin{proof}
The Monte Carlo estimator (6) has bias $\sum_{i=1}^N\delta_{i,time}(x_0),$ which is negligible under assumption 2; and variance 
$$\sum_{i=1}^N(\Delta t_i)^2\mathrm{Var}_{y(\tilde{t_i})}(\delta_{\tilde{t_i},space}(x_0,y(\tilde{t_i}))),$$
which, under assumption 1, can be approximated by
$$\int_0^Th(t)\mathrm{Var}_{y}(\delta_{_{t,space}}(x_0,y))dt,$$
where 
$h(t)=\sum_{i=1}^n\mathcal{X}_{[t_{i-1},t_i]}\Delta t_i$.\\
For a timestep scheme that partitions $[0,T]$ into $N$ intervals, we always have $\int_0^T\frac{1}{h(t)}dt=N$. Therefore, the calculus of variations shows that
 $$\mathrm{E}_{x_0\sim p_0}\int_{T_{trunc}}^Th(t)\mathrm{Var}_y(\delta_{t,space}(x_0,y))dt$$ 
attains its minimum when
$$\frac{1}{h(t)}=C\sqrt{\mathrm{E}_{x_0\sim p_0}\mathrm{Var}_y(\delta_{t,space}(x_0,y))},$$
and hence when the step size is chosen inversely proportionate to the square root of $\mathrm{E}_{x_0\sim p_0}\mathrm{Var}(\delta_{t,space}(x_0,y))$.
\end{proof}
\clearpage
\bibliographystyle{ACM-Reference-Format}
\bibliography{sample-base}
\end{document}